\documentclass[]{interact}

\usepackage{epstopdf}
\usepackage{color}
\usepackage{framed}
\usepackage{booktabs}
\usepackage{enumitem}

\usepackage[numbers,sort&compress]{natbib}
\bibpunct[, ]{[}{]}{,}{n}{,}{,}
\renewcommand\bibfont{\fontsize{10}{12}\selectfont}
\usepackage{amsthm}

\usepackage[labelfont=bf,font=small,labelsep=space]{caption}
\usepackage{subcaption}
\theoremstyle{plain}
\newtheorem{theorem}{Theorem}[section]
\newtheorem{lemma}[theorem]{Lemma}

\theoremstyle{definition}
\newtheorem{definition}[theorem]{Definition}

\theoremstyle{remark}

\renewcommand{\paragraph}[1]{\bigskip\noindent\textbf{#1.}}

\usepackage{multirow}
\usepackage{xspace}
\usepackage{bbm, dsfont}
\newcommand{\systemname}{Celer Network\xspace}

\newcommand{\tokenname}{CELR\xspace}
\newcommand{\orgname}{ScaleSphere Foundation Ltd. (``Foundation'')\xspace}

\newcommand{\routingname}{DBR\xspace}

\usepackage{setspace}
\usepackage[absolute]{textpos}

\begin{document}


\title{\begin{center}\systemname: Bring Internet Scale to Every Blockchain\end{center}}

\author{\begin{center}Mo Dong, Junda Liu, Xiaozhou Li, Qingkai Liang \\\orgname \end{center}}

\maketitle
\begin{spacing}{1.3}
\paragraph{Abstract} Just like how the 56Kbps dialup Internet in the 90s cannot possibly support 4K video streaming, the insufficient scalability of today's blockchain is the key factor limiting its use cases. Current blockchains have low throughput because each operation needs to be processed by the vast majority of nodes to reach on-chain consensus, which is exactly "how to build a super slow distribution system". Ironically, the on-chain consensus scheme also leads to poor privacy as any node can see the full transaction history of one another. While new consensus algorithms keep getting proposed and developed, it is hard to free on-chain consensus from its fundamental limitations.

Off-chain scaling techniques allow mutually distrustful parties to execute a contract locally among themselves instead of on the global blockchain. Parties involved in the transaction  maintain a multi-signature fraud-proof off-chain replicated state machine, and only resort to on-chain consensus when absolutely necessary (e.g., when two parties disagree on a state). Off-chain scaling is the only way to support fully scale-out decentralized applications ("dApps") with better privacy and no compromise on the trust and decentralization guarantees. It is the inflection point for blockchain mass adoption, and will be the engine behind all scalable dApps. 

Celer Network is an Internet-scale, trust-free, and privacy-preserving platform where everyone can quickly build, operate, and use highly scalable dApps. It is not a standalone blockchain but a networked system running on top of existing and future blockchains. It provides unprecedented performance and flexibility through innovation in off-chain scaling techniques and incentive-aligned cryptoeconomics.

Celer Network embraces a layered architecture with clean abstractions that enable rapid evolution of each individual component, including a generalized state channel and sidechain suite that supports fast and generic off-chain state transitions; a provably optimal value transfer routing mechanism that achieves an order of magnitude higher throughput compared to state-of-the-art solutions; a powerful development framework and runtime for off-chain applications; and a new cryptoeconomic model that provides network effect, stable liquidity, and high availability for the off-chain ecosystem.
\end{spacing}

\newpage
\setcounter{tocdepth}{3}
\tableofcontents
\newpage

\section{Introduction}
\label{sec:intro}
Many modern economic activities are essentially the flow and exchange of information and value. Over the past two centuries, the transfer of information has evolved from discrete events through pigeon networks to continuous flows through the speed-of-light Internet. However, the value transfer portion is far from light speed and is still very much discrete events controlled by segregated financial silos. This mismatch creates a devastating bottleneck in economic evolution: no matter how fast information flows, the expensive and slow value transaction is limiting the productive exchange of the two.

Essentially a revolutionary abstraction of trust among distrustful parties that results in an incentive-aligned distributed consensus, blockchain technology offers the foundation to dismantle segregated financial silos and dramatically expand the scope and freedom of global value flows. In practice, however, blockchain is deviating further away from the “speed-of-light” vision due to its low processing power compared to traditional value transfer tools. Scalability is a fundamental challenge that is hindering mass adoption of blockchain technology.

We envision a future with decentralized ecosystems where people, computers, mobile and Internet-of-Things ("IOT") devices can perform secure, private, and trust-free information-value exchange on a massive scale. To achieve this, blockchains should match the scale of the Internet and support hundreds of millions or billions of transactions per second. However, given the processing speed of existing blockchains (i.e., a few or tens of transactions per second), is it really possible to bring the scale of the Internet to blockchains? The answer is yes but only with off-chain scaling.

While on-chain consensus is the foundation of blockchain technology, its limitations are also obvious. In a sense, consensus is the opposite of scalability. For any distributed system, if all nodes need to reach consensus on every single transaction, its performance will be no better (in fact, much worse due to communication overhead) than a centralized system with a single node that processes every transaction, which means the system is eventually bottlenecked by the processing power of the slowest node. On-chain consensus also has severe implications on privacy, because all transactions are permanently public. A few on-chain consensus improvements have been proposed including sharding and various Proof-of-X mechanisms. They make the blockchain relatively faster with different tradeoffs in performance, decentralization, security, and finality, but cannot change the fundamental limitations of on-chain consensus.   

To enable Internet-scale blockchain systems with better privacy and no compromise on trust or decentralization, we have to look beyond on-chain consensus improvements. The core principle to design a scalable distributed system is to make operations on different nodes mostly independent. This simple insight shows that the only way to fully scale out  decentralized applications is to bring most of the transactions off-chain, avoid on-chain consensus as much as possible and use as a last resort. Related techniques include state channel, sidechain, and off-chain computing Oracle. Despite its high potentials, off-chain scaling technology is still in its infancy with many technical and economic challenges remaining unresolved. 

To enable off-chain scaling for prime-time use, we propose Celer Network\footnote{\emph{celer}: swift, fast in Latin, the \emph{c} for the speed of light}, a coherent architecture that brings Internet scale to existing and future blockchains. Celer Network consists of a carefully designed off-chain technology stack that achieves high scalability and flexibility with strong security and privacy guarantees, and a game-theoretical cryptoeconomic model that balances any new tradeoffs.

\subsection{Celer Technology Stack}
As a comprehensive full-stack platform that can be built upon existing or future blockchains, Celer Network encompasses a cleanly layered architecture that decouples sophisticated off-chain platform into hierarchical modules. This architecture greatly simplifies the system design, development, and maintenance, so that each individual component can easily evolve and adapt to changes. 

A well-designed layered architecture should have open interfaces that enable and encourage different implementation on each layer as long as they support the same cross-layer interfaces. Each layer only needs to focus on achieving its own functionality. Inspired by the successful layered design of the Internet, Celer Network adopts an off-chain technology stack that can be built on different blockchains, named \textbf{cStack}, which consists of the following layers in bottom-up order:

\begin{itemize}
\item \textbf{cChannel}: generalized state channel and sidechain suite. 
\item \textbf{cRoute}: provably optimal value transfer routing. 
\item \textbf{cOS}: development framework and runtime for off-chain enabled applications
\end{itemize}

\begin{figure}[t]
\begin{center}
\includegraphics[width=5.5in]{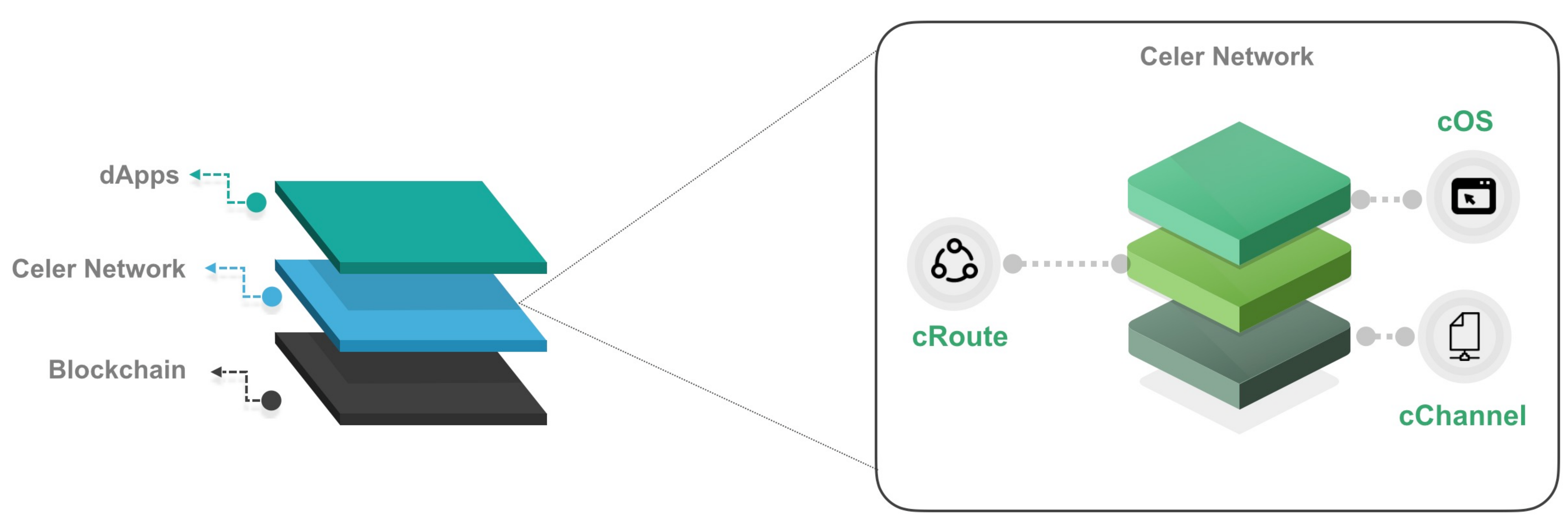}
\caption{Celer Network layered architecture.}
\label{fig:overall_arch}
\end{center}
\end{figure}

Celer architecture provides innovative solutions on all layers. Below we highlight the technical challenges and features of cChannel, cRoute, and cOS.

\paragraph{cChannel} This is the bottom layer of Celer Network that interacts with different underlying blockchains and provides the upper layer with a common abstraction of up-to-date states and bounded-time finality. cChannel uses state channel and sidechain techniques, which are both cornerstones of off-chain scaling platforms.

A state channel allows mutually distrustful parties to execute a program off-chain and quickly settle on the latest agreed states, with their security and finality guaranteed by on-chain bond contracts. It was initially introduced by Lightning Network~\cite{lightning} to support high-throughput off-chain Bitcoin micropayments. Since the introduction of Lightning Network, there have been several research works that addressed different problems in the context of payment channel networks, such as routing~\cite{flare, speedy, landmark}, time lock optimization~\cite{sprite}, and privacy~\cite{pripay}. However, off-chain network is still in its early stage, facing a few major challenges in terms of modularity, flexibility, and cost-efficiency. cChannel meets current challenges by offering a set of new features.
\begin{itemize}
\item \textbf{Generic off-chain state transition.} Off-chain transactions can be arbitrary state transitions with dependency DAG. This allows Celer Network to support complex high-performance off-chain dApps such as gaming, online auction, insurance, prediction market and decentralized exchanges.
\item \textbf{Flexible and efficient value transfer.} Multiple state channel and sidechain constructions with different tradeoffs on efficiency and finality are provided to support fast value transfer with generic condition dependency, minimal on-chain interactions, and minimal fund lockup.
\item \textbf{Pure off-chain contract.} Any contract that is not directly associated with on-chain deposits does not need any on-chain operation or initialization unless a dispute is triggered. Every pure off-chain contract or object has a uniquely identifiable off-chain address, and only needs to be deployed on blockchains when necessary with an on-chain address assigned by the built-in off-chain address translator.
\end{itemize}

\paragraph{cRoute} Celer Network is a platform for highly scalable dApps, designed to support high-throughput value transfer. Off-chain value transfer is an essential requirement of many off-chain applications. While Celer Network is capable of supporting dApps beyond payment solutions, it also makes groundbreaking improvements to off-chain payment routing as it directly determines how much and how fast value can be transferred within the ecosystem. 

All of the existing off-chain payment routing proposals~\cite{landmark, embedding, flare, silent, speedy, revive} boil down to conventional ``shortest path routing'' algorithms, which may achieve poor performance in an off-chain payment network due to the fundamental differences in the link model. The link capacity of a computer network is stateless and stable (i.e., not affected by past transmissions). However, the link capacity of an off-chain payment network is stateful (i.e., determined by on-chain deposits and past payments), which leads to a highly dynamic network where the topology and link states are constantly changing. This makes conventional shortest path routing algorithms hard to converge, and thus yields low throughput, long delay, and even outages.

To counter this fundamental challenge, Celer Network's payment routing module, cRoute, introduces \textbf{Distributed Balanced Routing (\routingname)}, which routes payment traffic using distributed congestion gradients. We highlight a few unique properties of \routingname (details in \S~\ref{sec:bp-description}). 
 
\begin{itemize}
\item \textbf{Provably optimal throughput.} We prove that for any global payment arrival rate, if there exists a routing algorithm that can support the rate, then \routingname can achieve that. Our evaluation shows that \routingname achieves 15x higher throughput and 20x higher channel utilization ratio\footnote{the ratio between the amount of transferred fund in each time slot and the total deposits of all channels.} compared to state-of-art solutions.
\item \textbf{Transparent channel balancing.} ``“Keeping the channels balanced'' has been our goal since the proposal of Lightning Network. However, existing attempts in channel balancing comprise heuristics that require heavy on-chain or off-chain coordination with poor guarantees. \routingname embeds the channel balancing process along with routing and constantly balances the network without requiring any additional coordination.
\item{\textbf{Fully decentralized}. The \routingname algorithm is a fully decentralized algorithm where each node only needs to "talk" to its neighbors in the state channel network topology. \routingname's  protocol also lowers messaging cost.}
\item{\textbf{Failure resilience}. The \routingname algorithm is highly robust against failures: it can quickly detect and adapt to unresponsive nodes, supporting the maximum possible throughput over the remaining available nodes.}
\item{\textbf{Privacy preserving}. The \routingname algorithm can be seamlessly integrated with onion routing~\cite{onion} to preserve anonymity for sources/destinations. Due to its multi-path nature, the \routingname algorithm  naturally preserves the privacy regarding the amount of transferred value, without using any additional privacy-preserving techniques (e.g., ZKSNARK).}
\end{itemize}


\paragraph{cOS} An on-chain dApp is simply a frontend connecting to the blockchain. Off-chain dApps, though with great potentials for high scalability, are not as easy to build and use as the traditional on-chain dApps. Celer Network introduces cOS which is a development framework and runtime for everyone to easily develop, operate, and interact with scalable off-chain dApps without being bogged down by the additional complexities introduced by off-chain scaling. Celer Network allows developers to concentrate on the application logic and create the best user experience, with the cOS dealing with the heavy lifting including the following tasks.

\begin{itemize}
\item Figure out the dependency between arbitrary off-chain and on-chain states. 
\item Handle the tracking, storage, and dispute of off-chain states.
\item Tolerate intermediate node failures transparently. 
\item Support multiple concurrent off-chain dApps.
\item Compile a unified implementation to different on-chain and off-chain modules. 
\end{itemize}

\subsection{Celer Cryptoeconomics}

Celer Network's cryptoeconomic mechanism, cEconomy, is designed based on a fundamental principle: a good cryptoeconomic (token) model should provide additional values and enable new game-theoretical dynamics that are otherwise impossible. While gaining scalability, an off-chain platform is also making tradeoffs on network liquidity and state availability, and it will never take off without a cryptoeconomic model that can enable new dynamics to balance out these tradeoffs.

\paragraph{New tradeoffs} Off-chain platform gains scalability by making the following tradeoffs.
\begin{itemize}
\item \textbf{Scalability vs. Liquidity.} Off-chain value transfer requires deposits to be locked on-chain as network liquidity. This is especially challenging for potential off-chain service providers, because a significant amount of liquidity is needed to provide effective off-chain services for global blockchain users, either as outgoing deposits in state channels or fraud penalty bond in sidechains. However, owners of a large number of crypto assets (whales) may not have the business interest or technical capability to run an off-chain service infrastructure, while people who have the technical capability of running a reliable and scalable off-chain service often do not have enough capital for channel deposits or fraud-proof bonds. Such a mismatch creates a huge hurdle for the mass adoption and technical evolution of off-chain operating networks.

\item \textbf{Scalability vs. Availability.} While off-chain scaling does not make any compromise on the trust-free property of the blockchain, it does sacrifice the availability guarantee. Each state channel or off-chain contract is associated with a dispute timeout, and the involved party will be at risk when staying off-line longer than the timeout, or when the local states are lost. 
\end{itemize}

Therefore, we need an incentive-compatible mechanism to provide sufficient liquidity for  entities which are capable of running a reliable and scalable off-chain service infrastructure, and to guarantee that the off-chain states are always available for possible on-chain dispute.

\paragraph{New cryptoeconomics}
To complete the off-chain scaling solution, we introduce a suite of cryptoeconomic mechanisms, named cEconomy, that brings indispensable value and provides network effect, stable liquidity, and high availability through the Celer Network's protocol token (``\tokenname") and three tightly coupled components.
\begin{itemize}
\item \textbf{Proof of Liquidity Commitment(PoLC).} PoLC is a virtual mining process that acquires abundant and stable liquidity for the off-chain ecosystem. To participate, one simply needs to commits (locks) his idle liquidity (in the form of digital assets, including but not limited to cryptocurrencies and CELR)  to the off-chain platform for a certain period of time with \tokenname rewarded as incentives to such users.  
\item \textbf{Liquidity Backing Auction (LiBA).} LiBA enables off-chain service providers to solicit liquidity through ``crowd lending" with negotiated interest rates. Lenders are ranked according to their ``happiness scores" that are determined by the interest rate, the amount of provisioned liquidity and the amount of staked \tokenname. In particular, lenders who stake more \tokenname (as an indicator for their past contributions to the ecosystems) have higher priority to be selected to provide liquidity to off-chain service providers.
\item \textbf{State Guardian Network (SGN).} SGN is a special compact sidechain that guards the states when users are offline so that the users' states are always available for dispute. Guardians need to stake their \tokenname into SGN to earn guarding opportunities and service fees from the users.
\end{itemize}
Section~\ref{sec:econ} introduces cEconomy mechanisms in detail with analysis of the \tokenname value and model incentive-compatibility.
\section{cChannel: the Foundation of Off-Chain Scaling}
\label{sec:channel}

\systemname's cChannel aims to provide a framework to enable a state channel and sidechain network with maximum flexibility and efficiency. This section starts with the generalized channel construct and outlines the key elements to support arbitrary conditional dependency between on-chain verifiable states. We then expand the horizon beyond classic state channel and examine how to encapsulate sidechain into the same interface exposed to the upper layer.

\subsection{Generalized State Channel}

\subsubsection{Key Idea and a Simple Example}
One major limitation of the existing payment network solutions is the lack of support for generalized state transitions. The need for generalized state transitions comes with the rise of smart contract platforms such as Ethereum. Smart contract enables asynchronous value transfer based on arbitrary contractual logic. To improve the scalability of such blockchains using off-chain state channel concepts, on-chain state transitions should be put into off-chain state channels and the corresponding value transfer should be made aware of such state transitions. 

We use a simple example of \emph{conditional payment} to illustrate the key idea about how to transform on-chain states transition to off-chain states transitions. Let's say Alice and Carl want to play a board game while betting on the result of such game in a trust-free manner: Alice will pay Carl \$1 if Carl wins and vice versa. 

This is a simple logic to implement on-chain. One could build a smart contract that holds Alice's and Carl's payment before the game starts. Alice and Carl will just play that game by calling the on-chain smart contract's functions. When one of them loses the game, surrenders, or times out, the winner gets the loser's deposit. The deposits can be seen as payments that are conditionally issued (i.e. condition on that the counterparty wins). Unfortunately, on-chain smart contract operations are extremely slow and expensive as every transaction involves an on-chain transaction.

Off-chain state channel can be used to  significantly improve scalability while maintaining the same semantic. Let's assume there is a payment channel between Alice and Carl. To enable the above semantic, we need to expand the functionality of channel's state proof to include a conditional lock that depends on the game's winner state. Alice can then send Carl an off-chain conditional payment effectively saying: ``I will pay Carl \$1 if the game contract's $who\_is\_winner$ function says Carl wins". The game state transitions can be also moved off-chain. The most straightforward way is to still have an on-chain contract governing the rule of the board game and that contract's address is referenced in the conditional payment. All the state transitions are happening off-chain via mutually signed game states that can be injected into the on-chain contract when necessary. 

But in fact, since there is no requirement for any kind of value bond for the program states, the entire game contract and the associated states can always stay off-chain as long as involved parties are collaborative. The only requirement is that the relevant game states are on-chain verifiable when need to be. An on-chain verifiable state means other contracts or objects can refer to it with no ambiguity. To realize that, we need to have a reference translator contract that maps off-chain references (such as the hash of contract code, constructor parameters, and nonce) to on-chain references (contract address). With these constructs, the game between Alice and Carl involves only one long-term on-chain contract that is not specific to the game logic, and no on-chain operation or initialization for the gaming. 

The example above reflects a specialized and simple instance of off-chain design patterns and it can be much more sophisticated. The conditional payment can be more complicated than just simple Boolean conditions, and can be designed in a way to redistribute locked up liquidity based on arbitrary contractual logic. In fact, \emph{conditional payments} is simply a special case of a more generalized \emph{conditional state transition}. The channel dependency can also be more complicated than one-to-one dependency to realize the common pattern of multi-hop state relays. We detail out the technical specification in the followings sections.

\subsubsection{Design Goals}
Our top goal is to achieve fast, flexible and trust-free off-chain interactions. We expect in most cases off-chain state transitions will stay off-chain until final resolution. Therefore, we aim to optimize commonly used off-chain patterns into succinct interactions with built-in support from on-chain components. 

Our second goal is to design data structure and object interaction logic that works for different blockchains. Celer Network aims to build a blockchain-agnostic platform and to run on different blockchains that support smart contracts. Therefore, a common data structure schema and a certain layer of indirection are required.

Besides these two highlighted goals, we plan to employ formal specification of channel state machines and verify the security properties along with the communication protocols that alter those states. We should also aim to provide an efficient on-chain resolution mechanism whenever possible.

\subsubsection{General Specification}

In this section, we provide specifications for the core components of cChannel's generalized state channel with a top-down approach, and describes the \textbf{Common State Channel Interface} that applies to any state channel with value transfer and arbitrary contractual logic. There could be extensive specialization and optimization for different concrete use cases, but the principles stay the same.

Before the detailed specification of a generalized state channel, we first introduce several important notations and terms that will be used throughout this section.
\begin{itemize}[leftmargin=*]
\item{(\textbf{State}). Denote by $s$ the state of a channel. For a bi-party payment channel, $s$ represents the  available balances of the two parties; for a board game, $s$ represents the board state.}
\item{(\textbf{State Proof}).  State proof serves as a bridge data structure between on-chain contracts and off-chain communication protocols. A state proof $sp$ contains the following fields
\begin{equation} \label{eq:state_proof}
    sp= \big\{\Delta s, seq, merkle\_root, sigs \big\},
\end{equation}
where $\Delta s$ denotes the accumulative state updates up until now. Note that given a base state $s_0$ and a state update $\Delta s$, we can uniquely produce a new channel state $s$. For example, in a bi-party payment channel, the base state $s_0$ corresponds to deposits of the two parties and the state update $\Delta s$ is a mapping that indicates the amount of tokens transferred from one participant to the other participant.
$seq$ is the sequence number for the state proof. A state proof with a higher sequence number will disable state proofs with lower sequence numbers. $merkle\_root$ is the root of the merkle tree of all pending condition groups and is crucial for creating conditional dependency between states in cChannel.  Finally, $sigs$ represents the signature from all parties on this state proof. State proof is valid only if all parties signatures are present. }
\item{(\textbf{Condition}). Condition $cond$ is the data structure representing the basic unit of conditional dependency and this is where the conditional dependency DAGs are weaved. A condition can be specified as follows.
\begin{equation} \label{eq:condition}
cond = \big\{timeout, *\textsc{IsFinalized}(args), *\textsc{QueryResult}(args) \big\}
\end{equation}
Here, $timeout$ is the timeout after which the condition expires. For example, for a condition that depends on the results of a board game, $timeout$ may correspond to the maximum duration of the board game (e.g., tens of minutes). Boolean function pointer $\textsc{IsFinalized}(args)$ is used to check whether the condition has been resolved and settled before the condition timeout. The arguments for this function call are application-specific. For example, in the board game, the arguments could be as simple as $args=[blocknumber]$ querying  whether the game winner has been determined before $blocknumber$. In addition, $\textsc{QueryResult}(args)$ is a result query function pointer that returns arbitrary bytes as the condition's resolving result. For example, in the board game, the arguments could be $args=[player1]$ querying whether $player1$ is the winner (boolean condition); in the second-price auction, the arguments could be $args=[participant1, participant2, \cdots, participantN]$  querying who is the winner and the amount of money each participant should pay (generic condition). The resolution process for a condition is to first perform $\textsc{IsFinalized}(args)$ and then perform result query $\textsc{QueryResult}(args)$.}
\item{(\textbf{Condition Group}). Condition group $cond\_group$ is a higher-level abstraction for a group of conditions to express generalized state dependencies. A condition group can be specified as follows.
\begin{equation} \label{eq:condition_group}
cond\_group = \big\{\Lambda, \textsc{ResolveGroup}(cond\_results) \big\},
\end{equation}
where $\Lambda$ denotes a set of conditions contained in this condition group. Each condition $cond \in \Lambda$ resolves to an arbitrary bytes array (i.e., the output of $cond.\textsc{QueryResult}(args)$). These bytes array are handled by a group resolving function $\textsc{ResolveGroup}(cond\_results)$ which takes the resolving results of all conditions as inputs and returns a state update $\Delta s$. For a payment channel, each condition group corresponds to a conditional payment. For example, a conditional payment saying that ``A pays B \$1 if B wins the Gomoku game'' corresponds to a condition group that contains two the conditions: the Hashed Time Lock condition (for multi-hop relay) and the Gomoku game condition (``B wins the game"). The $\textsc{ResolveGroup}$ function simply returns a transfer of \$1 from A to B if both conditions are true.
}
\end{itemize}

Now we are ready to specify the interface for a state channel. A state channel $\mathcal{C}$ can be specified as the following tuple:
\begin{equation} \label{eq:generalized_channel}
\mathcal{C} = \big\{p, s_0, sp, s, \mathcal{F}, \tau \big\},
\end{equation}
$p = \{p_1, p_2, ..., p_n\}$ is the set of participants in this channel. $s_0$ is the on-chain base state for this channel (e.g., initial deposits for each participant in a payment channel).  $sp$ represents the most updated known state proof for the channel.  $s$ is the updated channel state after state proof $sp$ is fully settled. $\tau$ is the settlement timeout increment for the state proof that will be specified later. $\mathcal{F}$ contains a set of standard functions that should be implemented by every state channel:
\begin{itemize}[leftmargin=*]
\item{$\textsc{ResolveStateProof}(sp,~cond\_groups)$. This function updates the current state proof by resolving attached condition groups.}
\item{$\textsc{GetUpdatedState}(sp,~s_0)$. This function is used to get the most updated state based on off-chain state proof $sp$ and on-chain base states $s_0$.}
\item{$\textsc{UpdateState}(s)$. This function allows on-chain updates of state channel's currently resolved state $s$.}
\item{$\textsc{IntendSettle}(new\_sp)$. This function opens a challenge period before the settlement timeout. During the challenge period, this function takes a state proof as input and update the current state proof if the input is newer.}
\item{$\textsc{ConfirmSettle}(sp)$. This function validates and confirms the current state proof as fully settled given the current time  has exceeded the settlement timeout.}
\item{$\textsc{IsFinalized}(args)$ and $ \textsc{QueryResult}(args)$ are the entry points for resolving  conditional dependency.  It accepts outside queries with necessary arguments for the querying contract to interpret accordingly. In fact, some patterns are used frequently enough, in cChannel's implementation, we separate them into pre-defined function interfaces.}
\item{$\textsc{CloseStateChannel}(s)$. This function terminates the life cycle of the state channel and distributes necessary states out according the latest settled state $s$.}
\end{itemize}
Settlement timeout is determined based on time of last called $\textsc{ResolveStateProof}$ or $\textsc{SettleStateProof}$ and the settlement timeout increment is $\tau$.

\paragraph{Dependency Constraints} When we create dependencies among different state channels, some constraints need to be enforced in order to guarantee a proper resolution of the dependency DAG. Suppose state channel $\mathcal{C}_1$ depends on state channel $\mathcal{C}_2$. Then it is required that the participants of $\mathcal{C}_1$ should be a subset of the participants of $\mathcal{C}_2$ such that the participants of $\mathcal{C}_1$ have the necessary information resolving its dependency on $\mathcal{C}_1$.

\subsubsection{Common Utilities}
The above abstraction defines the common pattern for generalized state channel construction. In different blockchains, the actual implementation might be different. For example, in Ethereum, cross-contract calls contains return value but in Dfinity, cross-contract calls only trigger registered callbacks. Reviewing multiple blockchains' implementations on state transition VMs, we identify two common utilities that are essential for the operation of generalized state channel in practice as followings:

\begin{itemize} [leftmargin=*]
\item \textbf{Off-chain Address Translator (OAT).} In the above abstraction, condition and condition group are associated with different functions. These functions should be the reference of on-chain contract's functions, but since program (smart contract) states are not inherently bound to constraints on the blockchain, there should be no fundamental requirement to have an on-chain presence. The only barrier to moving them entirely off-chain is the possible ambiguity of reference for functions such as $\textsc{IsFinalized}$ and $ \textsc{QueryResult}$. 

To resolve this ambiguity, we can define an on-chain rule set to map off-chain references to on-chain references. Off-chain Address Translator is built for that. For a contract with no value involved, it can be referenced by a unique identifier generated through its contract code, initial states, and a certain nonce. We call such unique identifier off-chain address. When resolving conditions on-chain, the referenced contracts need to be deployed and the corresponding functions (e.g., $\textsc{IsFinalized}$ and $ \textsc{QueryResult}$) should be able to translate off-chain addresses to on-chain addresses. To realize such functionality, OAT needs to be able to deploy contract code and initial states on-chain to get the on-chain contract and establish the mapping from off-chain address to on-chain address. 

\item \textbf{Hash Time Lock Registry(HTLR).} Hash Time Lock is commonly used in the scenario where transactions involving multiple state channels need to happen atomically. For example, multi-hop relayed payment (unconditional or conditional), atomic swap between different tokens, cross-chain bridges and more. HTL can be implemented entirely off-chain, but as Sprite~\cite{sprite} has pointed out, this is an over-optimization that actually limits the off-chain scalability. Therefore, Sprite~\cite{sprite} proposes a central registry where all locks can refer to. We extend and modify Sprite to fit in cChannel's common model. Effectively, HTLR provides dependency endpoints ($\textsc{IsFinalized}$, $\textsc{QueryResult}$) for conditions that act as locks. $\textsc{IsFinalized}$ takes a hash and block number and returns true if the corresponding pre-image has been registered before the block number. $\textsc{QueryResult}$ takes a hash and returns true if the pre-image of the hash is registered. These two functions can be simplified further into one, but for the sake of generality, we can simply keep them as two separate functions. Note that HTLR, and associated $\textsc{IsFinalized}$ and $\textsc{QueryResult}$, is always on-chain.
\end{itemize}

\subsubsection{Out-of-box Features}

In addition, we need to look into patterns that would be commonly used and enhance certain on-chain components with out-of-box features to simplify the corresponding off-chain interactions. \textbf{Generalized Payment Channel (GPC)} is a very good example of that. Generalized Payment Channel is payment channel that conforms to the general state channel specification and therefore can support various conditional payments based on further on-chain or off-chain objects. 

We first make the abstract model more concrete in the context of GPC. $s_0$ represents the static deposit map for each party in $p$. $s$ represents the final netted value each party owns. $\textsc{SubmitStateProof}$ is the function to submit a state proof and trigger a timeout challenge period before $\textsc{SettleStateProof}$ can be called and confirm the state proof. $\textsc{IsFinalized}$ and $\textsc{QueryResult}$ are functions to check whether the state of this payment channel has been finalized and query the current balances. One may wonder why a payment channel has an interface for outside query. This is because some other payments or states may depend on $s$ or existence of certain conditional payment locked in $sp$. $\textsc{ResolveStateProof}$ is the most interesting part as this is where a lot of specialized optimization will happen and greatly reduce the off-chain interaction complexity. $\textsc{GetUpdatedState}$ is a straightforward function to compute the netted out payment for each party based on the initial deposit and the fully resolved $sp$. $\textsc{CloseStateChannel}$ simply closes the channel and distributes the netted out the final balance for each party. 

With this basic model, we discuss how we can further optimize the GPC constructs to enable useful out-of-box features.

\begin{itemize}[leftmargin=*]
\item \textbf{Cooperative Settling} In most cases, counterparties of state channel applications are cooperative. As a result, it is added complexity and expense to go through the challenge period and then settle. Therefore, cChannel enables cooperative settling where counterparties not only sign the most recent state proof, but also sign the resulting signature to show the agreement that the state update described in that state proof is indeed the final state. With this, the number of transactions to settle a state proof can reduce from 2 to 1. 
\item \textbf{Single-transaction Channel Opening} Another  optimization cChannel brings is to reduce the number of on-chain operation to open a channel from 3 to 1. This is achieved by using a dependency contract to store deposits for counterparties. The $N-1$ counterparties will just sign the authorization of withdrawal entirely off-chain and one counterparty can submit that on-chain to complete the channel opening process. 
\item \textbf{Direct Final State Claim} When building generalized state channel applications, conditional state dependency is commonly used. When finalizing the GPC, one party may want to avoid the process of traversing the conditional dependency graph. This is to limit the griefing scenario of the counter-party, where the counterparty goes off-line and refuse to cooperatively convert some ConditionGroup to unconditional state update. To limit the work needed for the disputing party, we introduce the method of direct final state claim. It allows the online party to directly claim a final state without actually performing any additional traversal of dependency graph. No counterparty signatures are needed. To avoid abuse, a fraud-proof bond is also required for the claiming party. After a challenging period, the state will become final without needing to perform any additional operations. 
\item \textbf{Dynamic Deposit and Withdrawal.} A common requirement for GPC is to enable seamless on-chain transactions when the counterparty is not connected to the network. For withdrawals of funds, we introduces a pair of functions $\textsc{IntendWithdraw}$ and $\textsc{ConfirmWithdraw}$, to meet this requirement. $\textsc{IntendToWithdraw}$ changes the base state  $s_0$ with a challenge period. Counter-party can submit conflicting $sp$ to dispute. If no dispute happens before the challenge period defined by $\tau$, $\textsc{ConfirmToWithdraw}$ is called to confirm and distribute the withdrawal. These two functions work very much like $\textsc{IntendSettle}$ and $\textsc{ConfirmSettle}$. Deposit is straightforward as it only changes the base state $s_0$. 

\item \textbf{Boolean Circuit Condition Group.} We expect that GPC's most common use case would be Boolean circuit based conditional payment. For example, ``A pays B if function X or function Y return true''. To optimize for such payment, we tweak the interface of condition group and condition. In particular, we can specialize function $\textsc{ResolveGroup}$ to release a pre-defined conditional payment if any of the condition resolving results (or any boolean circuit of condition results) holds true. This way, we saved the trouble of creating additional objects for $\textsc{ResolveGroup}$ and the corresponding multi-party communication overhead. We also specify condition as Boolean condition so that we require the depending objects should have an interface with the effect of ``isSatisfied'' that returns true or false based on the state queried. 
\item \textbf{Fund Assignment Condition Group.} Another more generalized use case for GPC is generalized state assignment. We implement this by introducing another different type of condition group, which only has one condition in it. $\textsc{QueryResult}$ will directly return a state assignment map dictating an update for $\Delta s$. This enables a more general plug-in point for GPC. One can plugin an off-chain contract that was initialized with certain locked in liquidity. This contract can check not only who wins a game (Boolean) but how many steps the winner took to win the game, and then assign the liquidity by carrying out a certain computation. The involved parties can generate a Condition Group that references the check function of the off-chain contract address they mutually agree on. 
\end{itemize}

There can be many more common patterns defined for different patterns, but the above example illustrates the design principle for such optimization.

\subsection{Alternative Channel Model with Sidechains}
Besides the above-mentioned generalized state channel model, cChannel also introduces an alternative state channel model that is facilitated by sidechains~\cite{plasma}. 
For example, consider the scenario where multiple users need to pay each other. Users can pool their deposits to a central contract, which acts like a sidechain contract with the off-chain service providers playing the role of block proposers (forming a ``multi-party hub'' with off-chain service providers being ``hub operators''), and therefore enables one-to-many payment relationships within a hub. The integrity of the off-chain service providers is ensured by a certain level of fraud-proof bond acceptable by participants. 

Specifically, in \systemname, each off-chain service provider can run a sidechain-aided state channel:
\begin{equation} \label{eq:sidechain}
    \mathcal{C} = \{s, p^*, b, \tau \},
\end{equation}
where $s$ is the sidechain state,  $p^*$ is the single block proposer (the off-chain service provider), $b$ is the fraud bond and $\tau$ represents the finality timeout. Each node $i$ can send sidechain transactions to every other node in the channel to update the latest state. Just like any sidechain transactions, node $i$ will not only sign this transaction but also sign another transaction to prove that it has seen this transaction included in the block created by $p^*$. The second signed transaction can be used as proof for node $i$. As long as participants have block data fully available, the finality of this sidechain transaction can be confirmed quickly as well. 

This sidechain-aided channel model comes with the following expected benefits~\cite{plasma} when compared to previously mentioned channel models.
\begin{itemize}
\item \textbf{No on-chain transaction and online presence needed for the receiver.} This is a natural benefit inherited from sidechain properties. The reason is that the receiver can redeem their fund received on the sidechain-aided channel without actually performing any sidechain deposit themselves.

\item \textbf{No per-party fund lock-up.} This benefit is in the context of payment channels. When side-chain based channels are used for multi-party payment, each party does not need to lock their deposit in advance before they pay each other (except for the block proposer who needs to deposit fraud-proof bond).
\end{itemize}

However, the ecosystem should be clearly aware of the downsides of this channel model as the following:
\begin{itemize}
\item \textbf{Fraud proof bond is still needed.} In the case of sidechain based channels, the fraud-proof bond is still needed from either the block proposer $p^*$ directly or whoever is providing auditing and insurance services. It should be clearly understood that the worst-case liquidity requirement for the block proposer (i.e., the off-chain service provider) is actually unbounded. The reason is that with enough colluding, malicious party can create unbounded repeated spending.
\item \textbf{Data availability issue may complicate finality.} Even if no malicious parties are involved, the inherent finality delay for sidechain model still haunts this channel model, especially when data availability becomes an issue. When the block data is not always available among relevant parties, the sidechain faces inevitable re-organization and therefore finality will be delayed at the best and the entire sidechain will be abandoned in the worst case.
\end{itemize}

These sidechain based channels can be further connected to each other via the common state channels. 
\section{cRoute: Provably-Optimal Value Transfer Routing}
\label{sec:routing}
\subsection{Challenges in State Channel Network Routing}
The need for state routing (or ``payment routing" in the case of payment channel networks) is apparent: it is impractical to establish direct state channels between every pair of nodes due to channel opening costs and deposit liquidity lockup. Therefore, it is necessary to build a network consisting of state channels, where state transitions should be relayed in a trust-free manner. The design of state routing is crucial for the level of scalability that a state channel network can provide, i.e., how fast and how many transactions can flow on a given network. However, existing proposals all fall short to meet the fundamental challenges imposed by the unique properties of state channel networks.

Landmark routing \cite{landmark} has been proposed as one option for decentralized payment routing in several payment channel networks. For example,  Lightning Network \cite{lightning} adopts a landmark routing protocol called Flare \cite{flare}. A similar algorithm is also used in the decentralized IOU credit network SilentWhispers \cite{silent}. The key idea of landmark routing is to determine the shortest path from sender to receiver through an intermediate node, called a landmark, usually a well-known node with high connectivity.

Raiden Network \cite{raiden} (a payment channel network) mentioned a few implementation alternatives for payment routing, such as $A^*$ tree search which is a distributed implementation of shortest path routing. In addition, since route discovery is hard but crucial, nodes can provide pathfinding services for other nodes for some convenience fees in Raiden Network.

Recently proposed SpeedyMurmurs \cite{speedy} enhances the previous shortest path routing algorithms (as used in Lightning Network and Raiden Network) by accounting for the available balances in each payment channel.  Specifically, SpeedyMurmurs is based on the embedding-based routing algorithms that are commonly used in P2P networks \cite{embedding}, which first constructs a prefix tree and then assigns  a coordinate for each node. The forwarding of each payment is based on the distances between the coordinate known to that node and the coordinate of the destination. The prefix tree and the coordinate of each node will be adjusted if there is any link that needs to be removed (i.e., when a link runs out of balance) or added (i.e., when a depleted link receives new funding).

\begin{figure}[t]
\begin{center}
\includegraphics[width=5.5in]{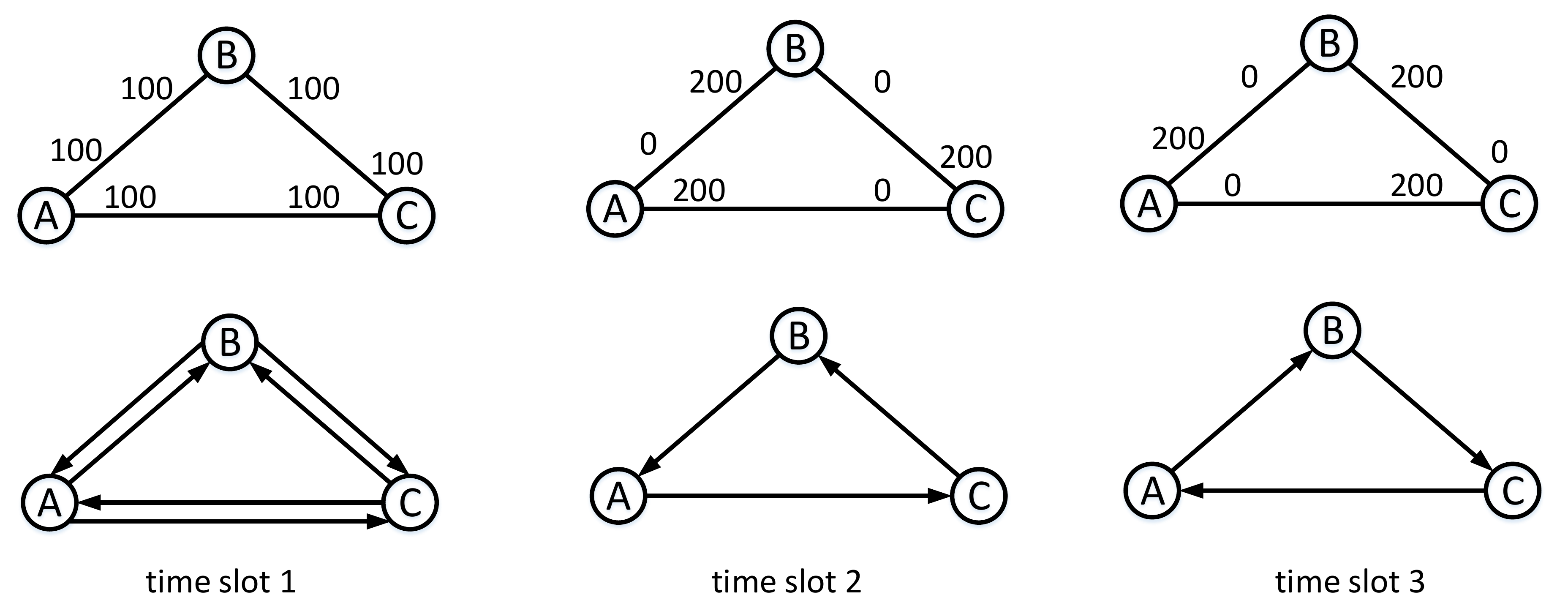}
\caption{Shortest path routing leads to frequent topology changes due to channel imbalance.}
\label{fig:sp-example}
\end{center}
\end{figure}

It can be observed that all of the existing routing mechanisms boil down to ``\emph{shortest path routing with available balance consideration}''. In traditional data networks, shortest path routing does provide reasonably good throughput and delay performance, based on the assumption that network topology remains relatively stable and link capacity is ``stateless" (i.e., the capacity of each link is not affected by past transmissions). Unfortunately, such an assumption no longer holds for an off-chain state channel network due to its ``stateful" link model, i.e., the capacity (available balance) of each directed link keeps changing as payments go through that link. Note that shortest path routing does not account for channel balancing, and thus each link may quickly run out of its capacity, which further leads to frequent changes in network topology. Figure \ref{fig:sp-example} illustrates a scenario in which shortest path routing leads to topology changes every time slot. Suppose that at the beginning of each time slot, node $A$, node $B$ and node $C$ each initiates a payment of 100 tokens to node $B$, node $C$ and node $A$, respectively. Under the initial channel balance distribution (time slot 1), every pair of nodes are connected by a bi-directional link, and  each node selects a direct route to its destination under shortest path routing. However, this results in a uni-directional transfer over each channel, and thus the distribution of channel balances becomes highly skewed in time slot 2, where the underlying topology is a counter-clockwise cycle. In this new topology, shortest path routing continues to make uni-directional transfers (e.g., selects route $A\rightarrow C\rightarrow B$ for payment $A\rightarrow B$), and channel balances are pushed to the other extreme, where the underlying topology is completely reversed to a clockwise cycle (time slot 3). The same pattern will repeat indefinitely. In contrast, if node $C$ takes a longer route $C\rightarrow B\rightarrow A$, every channel will remain balanced all the time, and the network topology never changes. For any decentralized implementation of shortest path routing, such frequent topology changes could lead to poor performance since it takes time for the algorithm to converge on the new topology (e.g., to reconstruct the prefix tree as in SpeedyMurmurs \cite{speedy}), during which sub-optimal routes may be taken. What's even worse is that the network topology may change again before the algorithm converges, and thus the algorithm may \emph{never converge} and achieve continually poor throughput performance.

Note that the recent project Revive \cite{revive} proposes an explicit channel rebalancing scheme. However, Revive does not account for state routing, which means that its channel rebalancing procedure is not transparent to the underlying routing process and requires extra out-of-band coordination. Moreover, Revive only works in a restricted class of network topologies that contain cyclic structures, and it does not provide any guarantee that its channel rebalancing procedure is feasible in a general topology. In comparison, we propose a routing algorithm that achieves transparent and optimal channel balancing during the routing process.



\subsection{Distributed Balanced Routing (DBR)}
We propose \emph{Distributed Balanced Routing} (DBR) as an efficient routing protocol for value transfers in an off-chain state channel network. The \routingname algorithm is inspired by the BackPressure routing algorithm \cite{BP1, BP2} that was originally used in wireless networks. It is based on a completely different design philosophy from the traditional shortest path routing.  In particular, \routingname does not perform any explicit path computation from source to destination. Instead, the routing direction is guided by the current network's \emph{congestion gradients}. Think of water flowing from the top of a hill to a destination at the foot of the hill. The water does not need to know the route to its destination; all it needs to do is to follow the direction of gravity.

The \routingname algorithm uses a similar design philosophy but also accounts for the stateful link model in state channel networks. In particular, the \routingname algorithm is augmented with a state channel balancing ability that transparently maintains balanced transfer flows for each state channel. Compared with existing routing algorithms, the proposed \routingname  algorithm has the following advantages:
\begin{itemize}
\item{\textbf{Provably optimal throughput}. In other words, for a given arrival rate of value transfer requests, if there exists any routing algorithm that ``supports" the rate, \routingname is also able to do that. The meaning of ``support" will be specified in Section \ref{sec:bp-performance}.}
\item{\textbf{Transparent channel balancing}.    In \routingname, the channel rebalancing process is naturally embedded in the routing process without any additional coordination. It automatically rebalances each state channel to maintain balanced value transfers in the long term.}
\item{\textbf{Fully decentralized}. The \routingname algorithm is a fully decentralized algorithm where each node only needs to talk to its neighbors in the state channel network topology. \routingname  also has low messaging cost in the protocol.}
\item{\textbf{Failure resilience}. The \routingname algorithm is highly robust against failures: it can quickly detect and adapt to unresponsive nodes, supporting the maximum possible throughput over the remaining available nodes.}
\item{\textbf{Privacy preserving}. Due to its multi-path nature, the \routingname algorithm  naturally preserves the privacy regarding the amount of transferred values, without using any additional privacy-preserving techniques (e.g., ZKSNARK). More importantly, the \routingname algorithm can be seamlessly integrated with onion routing~\cite{onion} to preserve anonymity for sources/destinations.}
\end{itemize}

In the following, we first introduce the state channel network model, then describe the \routingname algorithm, and finally prove the performance of \routingname. Note that for the ease of exposition, we restrict our attention to bi-party payment channels in this section, but the same ideas apply to any state channel network that has value transfer requirements.

\subsubsection{System Model} \label{sec:route:sysmodel}
In our model, time is discretized into slots of a fixed length, where the length of each slot usually corresponds to the physical transmission delay over one hop. Suppose that there are $N$ nodes in the network. For each pair of nodes $i$ and $j$, a pair of directed links $i\rightarrow j$ and $j \rightarrow i$ exit if there is a bi-directional payment channel $i\leftrightarrow j$ between node $i$ and node $j$.  Let $c_{ij}(t)$ be the capacity of link $i\rightarrow j$ in slot $t$, which corresponds to the remaining balance in the payment channel that can be transferred from node $i$ to node $j$ at the beginning of that slot. There is a total deposit constraint for each bi-directional payment channel between node $i$ and node $j$:
\[
c_{ij}(t) + c_{ji}(t) = B_{i\leftrightarrow j}(t),
\]
where $B_{i\leftrightarrow j}(t)$ is the total deposit of bi-directional payment channel $i\leftrightarrow j$ at the beginning of slot $t$. Note that the total deposit $B_{i\leftrightarrow j}(t)$ may change over time due to dynamic on-chain fund deposit/withdrawal.

During each slot $t$, each node $i$ receives new payment requests from outside the network, where the total amount of tokens that should be delivered to node $k$ is $a_i^{(k)}(t)\ge 0$. Also denote by $\mu_{ij}^{(k)}(t)$ the amount of tokens (required to be delivered to node $k$) sent over link $i\rightarrow j$ in slot $t$, which is referred to as a \emph{routing variable}.

\subsubsection{Protocol Description}\label{sec:bp-description}
Before the description of \routingname, we first introduce several important notions: debt queue, channel imbalance and congestion-plus-imbalance (CPI) weight.

\vspace{2mm}

\noindent \textbf{(Debt Queue).} In the operation of \routingname, each node $i$ needs to maintain a ``debt queue"  for payments destined to each node $k$, whose queue length $Q_i^{(k)}(t)$ corresponds to the amount of tokens (with destination $k$) that should be relayed by node $i$ to the next hop but have not been relayed yet at the beginning of slot $t$. Intuitively, the length of the debt queue is an indicator of congestion over each link. The queue length evolution is as follows:
\begin{equation}\label{eq:queue}
Q_i^{(k)}(t+1)=\Big[Q_i^{(k)}(t)+a_i^{(k)}(t)+\sum_{j\in \mathcal{N}_i} \mu_{ji}^{(k)}(t)-\sum_{j\in\mathcal{N}_i}\mu_{ij}^{(k)}(t)\Big]^+,
\end{equation}
where $[x]^+=\max\{0, x\}$ (since queue length cannot be negative) and $\mathcal{N}_i$ is the set of neighbor nodes of node $i$. The above equation simply means that in slot $t$, the change of queue length is caused by three components: (1) new token transfer requests  from outside the network (i.e., $a_i^{(k)}(t)$), (2) tokens routed from neighbors to node $i$, i.e., $\sum_{j\in \mathcal{N}_i} \mu_{ji}^{(k)}(t)$, and (3) tokens routed from node $i$ to its neighbors, i.e., $\sum_{j\in\mathcal{N}_i}\mu_{ij}^{(k)}(t)$. It should be noted that the queue length at the destination node is always zero, i.e., $Q_i^{(i)}(t)=0$ for each node $i$, which guarantees that every packet can be eventually delivered to its destination under the \routingname algorithm.

\vspace{2mm}

\noindent \textbf{(Channel Imbalance).} For each link $i\rightarrow j$, we define the channel imbalance as  
\begin{equation}
\Delta_{ij}(t)=\sum_{\tau<t}\sum_k \Big(\mu_{ji}^{(k)}(\tau)-\mu_{ij}^{(k)}(\tau)\Big).
\end{equation}
Intuitively, $\Delta_{ij}(t)$ is the difference between the total amount of tokens received by node $i$ from node $j$ and the total amount of tokens sent from $i$ to $j$ over their payment channel up to the beginning of slot $t$. Note that if $\Delta_{ij}(t)<0$ then it means that node $i$ sent more tokens to node $j$ than what was received from node $j$. Clearly, $\Delta_{ij}(t)$ is a natural measure of channel imbalance as perceived by node $i$. Our \routingname algorithm tries to balance the payment channel such that $\lim_{t\rightarrow \infty} \Delta_{ij}(t)\slash t=0$ for each payment channel $i\leftrightarrow j$, which implies that the long-term sending rate from $i$ to $j$ equals the sending rate from $j$ to $i$.

\vspace{2mm}

\noindent \textbf{(Congestion-Plus-Imbalance (CPI) Weight).} Define the Congestion-Plus-Imbalance (CPI) weight for link $i\rightarrow j$ and destination $k$ as
\begin{equation}
W^{(k)}_{ij}(t) = Q_{i}^{(k)}(t) - Q_{j}^{(k)}(t) + \beta\Delta_{ij}(t),
\end{equation}
where $\beta > 0$ is a parameter adjusting the importance of channel balancing. 
Intuitively, the above weight is the sum of the differential backlog $Q_{i}^{(k)}(t) - Q_{j}^{(k)}(t)$ for payments destined to node $k$ between node $i$ and node $j$ (i.e., congestion gradient) and the channel imbalance $\Delta_{ij}(t)$ between node $i$ and node $j$. The former is used to reduce network congestion and improve network throughput while the latter is used to balance payment channels.

\begin{framed}
\noindent  \textbf{Distributed Balanced Routing (\routingname)}

\vspace{3mm}

\noindent  The following protocol is locally executed by each node $i$.

\vspace{3mm}

In each time slot $t$, node $i$ first exchanges the queue length information with its neighbors and calculates the CPI weights. Then for each link $i\rightarrow j$, node $i$ calculates the best payment flow to transmit over that link:
\begin{equation}
k^*=\arg\max_k~~W^{(k)}_{ij}(t).
\end{equation}
If $W^{(k^*)}_{ij}(t) > 0$, then $\mu_{ij}^{(k^*)}(t)=c_{ij}(t)$ otherwise  $\mu_{ij}^{(k^*)}(t)=0$. For any $k\ne k^*$, set $\mu_{ij}^{(k)}(t)=0$. 

\end{framed}

\vspace{2mm}

\noindent \textbf{Remark.} In each slot $t$, \routingname  essentially tries to solve the following weighted-sum optimization problem:
\begin{equation}\label{eq:maxweight}
\begin{split}
\max ~~ &\sum_{ij}\sum_{k}\mu_{ij}^{(k)}(t)W_{ij}^{(k)}(t) \\
 \text{s.t.}~~ & \sum_{k} \mu_{ij}^{(k)}(t) +  \sum_{k} \mu_{ji}^{(k)}(t)\le B_{i\leftrightarrow j}(t),~\forall i,j.
 \end{split}
\end{equation}
The above optimization problem is also called MaxWeight and derived from our theoretical analysis of \routingname (see the next section). The aforementioned algorithm description gives an approximate solution to MaxWeight.

\subsubsection{Throughput Performance of \routingname }\label{sec:bp-performance}

To analyze the throughput performance of \routingname , we first introduce a few definitions. 

\begin{itemize}
\item{A state channel network is said to be \textbf{stable} if 
\[
\lim_{t\rightarrow \infty} \frac{Q_i^{(k)}(t)}{t} = 0,~\forall i,k,
\]
which implies that the long-term arrival rate to each debt queue equals the long-term departure rate from that queue.}
\item{A state channel network is said to \textbf{balanced} if 
\[
\lim_{t\rightarrow \infty} \frac{\Delta_{ij}(t)}{t}=0,~\forall \text{ channel }i\leftrightarrow j.
\]
In other words, for each payment channel $i\leftrightarrow j$ the long-term sending rate from node $i$ to node $j$ equals the sending rate from node $j$ to node $i$.}
\item{Define 
\[
\lambda_i^{(k)}\triangleq\lim_{t\rightarrow\infty}\frac{1}{t}\sum_{\tau=0}^{t-1}a_i^{(k)}(\tau)
\]
as the long-term average arrival rate to node $i$ for payments with destination $k$. An arrival rate vector $\bm{\lambda}=(\lambda_i^{(k)})_{i,k}$ is said to be \textbf{supportable} if there exists a routing algorithm that can keep the network stable and balanced under this arrival rate vector.}
\item{The \textbf{throughput region} of a state channel network is the set of supportable arrival rate vectors.}
\item{A routing algorithm is \textbf{throughput-optimal} if it can support any payment arrival rate vector inside the throughput region.}
\end{itemize}

For the ease of exposition, we assume that the external payment arrival process $\{a_i^{(k)}(t)\}_{t\ge 0}$ is stationary and has a steady-state distribution, and that the total deposit for each payment channel remains fixed, i.e., $B_{i\leftrightarrow j}(t)=B_{i\leftrightarrow j}$ for any $t\ge 0$. Our analysis can be extended to the case where the arrival process is non-stationary and channel deposits are time-varying (e.g., dynamic on-chain deposit/withdraw), at the expense of unwieldy notations. The following theorem shows the throughput performance of \routingname.
\begin{theorem}\label{thm:bp}
The \routingname algorithm is throughput-optimal.
\end{theorem}
\noindent In other words, as long as there exists a routing algorithm that can keep the payment network stable and balanced, the \routingname  algorithm can also achieve that. The rest of \S\ref{sec:bp-performance} in below is the proof of Theorem~\ref{thm:bp}. We first introduce a lemma which characterizes the throughput region for a state channel network.

\begin{lemma}\label{lm:throughput-region}
An arrival rate vector $\bm{\lambda}=(\lambda_i^{(k)})_{i,k}$ is supportable if and only if there exist flow variables $\bm{f}=(f_{ij}^{(k)})_{i,j,k}$ that satisfy the following conditions:
\begin{align}
&\lambda_i^{(k)} + \sum_{j\in\mathcal{N}_i} f_{ji}^{(k)} - \sum_{j\in\mathcal{N}_i} f_{ij}^{(k)}\le 0,~\forall k,i\ne k \label{eq:flow-conservation}\\
& \sum_k f_{ij}^{(k)} = \sum_k f_{ji}^{(k)},~\forall i,j \label{eq:channel-balance}\\
& \sum_k f_{ij}^{(k)} + \sum_k f_{ji}^{(k)} \le B_{i\leftrightarrow j},~\forall i,j. \label{eq:channel-capacity}
\end{align}
\end{lemma}
\begin{proof}
The necessity of the above conditions is trivial. Inequality \eqref{eq:flow-conservation} corresponds to the flow conservation requirement. If it is violated, then the arrival rate to node $i$ is larger than the departure rate, and the state channel network is unstable. Equation \eqref{eq:channel-balance} corresponds to the channel balance requirement. If it is violated, then channel $i\leftrightarrow j$ is imbalanced. Inequality \eqref{eq:channel-capacity} corresponds to the channel capacity constraint, since the sum of tokens transferred over each channel $i\leftrightarrow j$ cannot exceed the total channel deposit $B_{i\leftrightarrow j}$.

In order to prove the sufficiency of the above conditions, we construct an algorithm that can stabilize and balance the state channel network when the arrival rate vector $\bm{\lambda}$ satisfies \eqref{eq:flow-conservation}-\eqref{eq:channel-capacity}. The algorithm is straightforward: in each slot $t$, set the routing variable $\mu_{ij}^{(k)}(t)=f_{ij}^{(k)}$ for any $i,j,k$. Clearly, under this routing algorithm every channel $i\leftrightarrow j$ remains balanced in each slot $t$ since $\sum_k \mu_{ij}^{(k)}(t) = \sum_k f_{ij}^{(k)} = \sum_k f_{ji}^{(k)} = \sum_k \mu_{ji}^{(k)}(t)$. Moreover, the network is stable under the algorithm since the flow conservation requirement is satisfied for every node. Note that each link $i\rightarrow j$ may have insufficient fund initially (i.e., $c_{ij}(0) < \sum_k f_{ij}^{(k)}$) such that the routing decision is infeasible. In this case, we can let node $j$ transfer some tokens to node $i$ at the beginning in order to equalize the fund at both ends of the state channel.  Such an adjustment process incurs at most $B_{i\leftrightarrow j}$ sub-optimal transfers for each channel $i\leftrightarrow j$ and does not influence network stability and channel balance in the long term.

Therefore, equations \eqref{eq:flow-conservation}-\eqref{eq:channel-capacity} are a necessary and sufficient condition for an arrival rate vector $\bm{\lambda}$ to be supportable.
\end{proof}

It should be noted that the routing algorithm mentioned in the proof of Lemma \ref{lm:throughput-region} cannot be implemented in practice since we do not know the external  arrival rate vector $\bm{\lambda}$ in advance. In the following, we prove that \routingname  can achieve the same throughput performance without knowing any payment traffic statistics in advance.

\vspace{3mm}

By Lemma \ref{lm:throughput-region}, if an arrival rate vector $\bm{\lambda}$ belongs to the throughput region, it must satisfy \eqref{eq:flow-conservation}-\eqref{eq:channel-capacity} and can be supported by the algorithm mentioned in the proof of Lemma \ref{lm:throughput-region} (referred to as the \emph{optimal oracle algorithm}). In the following, denote by $\widetilde{\mu}_{ij}^{(k)}(t)$ the routing decision made by the optimal oracle algorithm in slot $t$. By the nature of the optimal oracle algorithm, we have $\widetilde{\mu}_{ij}^{(k)}(t)=f_{ij}^{(k)}$ for any $t$ (if ignoring the initial fund adjustment process). 


Define the \emph{Lyapunov function} as follows:
\begin{equation}\label{eq:potential}
\Phi(t)=\sum_{i,k}\Big(Q_i^{(k)}(t)\Big)^2 +  \frac{\beta}{2}\sum_{i,j}  \Delta_{ij}^2(t).
\end{equation} 
Also define the \emph{conditional Lyapunov drift} as $D(t)\triangleq \mathbb{E}[\Phi(t+1)-\Phi(t)|\mathbf{Q}(t),\Delta(t)]$, where the expectation is with respect to the randomness of arrivals. To facilitate the analysis, we assume that the amount of new payments arrivals to the network in each slot is bounded by some constant. By equation \eqref{eq:queue}, we have
\begin{equation}\label{eq:first-term}
\begin{split}
\Big(Q^{(k)}_i(t+1)\Big)^2&= \Big[Q_i^{(k)}(t)+a_i^{(k)}(t)+\sum_{j} \mu_{ji}^{(k)}(t)-\sum_{j}\mu_{ij}^{(k)}(t)\Big]^2\\
&= \Big(Q_i^{(k)}(t)+\sum_{j} \mu_{ji}^{(k)}(t)-\sum_{j}\mu_{ij}^{(k)}(t)\Big)^2 + \Big(a_i^{(k)}(t)\Big)^2 \\
&~~ + 2a_i^{(k)}(t)\Big(Q_i^{(k)}(t)+\sum_{j} \mu_{ji}^{(k)}(t)-\sum_{j}\mu_{ij}^{(k)}(t)\Big)\\
& \le  \Big(Q^{(k)}_i(t)\Big)^2 - 2Q^{(k)}_i(t)\Big(\sum_{j} \mu_{ij}^{(k)}(t)-\sum_{j}\mu_{ji}^{(k)}(t)\Big) \\
&~~ + 2a_i^{(k)}(t)Q_i^{(k)}(t) + const,
\end{split}
\end{equation}
where the inequality is due to the assumption that the arrival $a_i^{(k)}(t)$ in each slot $t$ is bounded by some constant and the fact that the number of transferred tokens in each slot is also bounded (since $\mu_{ij}^{(k)}(t)\le B_{i\leftrightarrow j}$). Now we have
\[
\begin{split}
&\sum_{i,k}\Big(Q^{(k)}_i(t+1)\Big)^2 - \sum_{i,k}\Big(Q^{(k)}_i(t)\Big)^2 \\
\le &~const - 2\sum_{i,k}Q^{(k)}_i(t)\Big(\sum_{j} \mu_{ij}^{(k)}(t)-\sum_{j}\mu_{ji}^{(k)}(t)\Big) + 2\sum_{i,k}a_i^{(k)}(t)Q_i^{(k)}(t)\\
=&~const - 2\sum_{i,j}\sum_{k}\mu_{ij}^{(k)}(t)\Big(Q_{i}^{(k)}(t)-Q_j^{(k)}(t)\Big) + 2\sum_{i,k}a_i^{(k)}(t)Q_i^{(k)}(t),
\end{split}
\]
where we rearrange the sum in the above equality. Similarly, noticing that 
\[
\Delta_{ij}(t+1)=\Delta_{ij}(t)+\sum_k \mu_{ji}^{(k)}(t)-\sum_k \mu_{ij}^{(k)}(t),
\]
we can prove
\begin{equation}\label{eq:second-term}
\begin{split}
\frac{\beta}{2}\sum_{i,j} \Delta_{ij}^2(t+1)-\frac{\beta}{2}\sum_{i,j}\Delta_{ij}^2(t) &\le \beta \cdot const - \beta\sum_{i,j}\sum_k\Delta_{ij}(t)\Big( \mu_{ij}^{(k)}(t)- \mu_{ji}^{(k)}(t)\Big)\\
& = \beta \cdot const - \beta\sum_{i,j}\sum_k\mu^{(k)}_{ij}(t)\Big( \Delta_{ij}(t)- \Delta_{ji}(t)\Big)\\
& = \beta \cdot const - 2\sum_{i,j}\sum_k\mu^{(k)}_{ij}(t)\beta \Delta_{ij}(t),
\end{split},
\end{equation}
where we use the fact that $\Delta_{ij}(t)=-\Delta_{ji}(t)$. As a result, by combining \eqref{eq:first-term} and \eqref{eq:second-term},  the conditional Lyapunov drift can be bounded by
\[
\begin{split}
D(t) &\le \beta c_1 + c_2 - 2\sum_{i,j}\sum_k\mu^{(k)}_{ij}(t)\Big(Q_{i}^{(k)}(t)-Q_j^{(k)}(t)+\beta  \Delta_{ij}(t)\Big) + 2\sum_{i,k}\lambda_i^{(k)}Q_i^{(k)}(t)\\
& \le \beta c_1 + c_2 - 2\sum_{i,j}\sum_k  \widetilde{\mu}^{(k)}_{ij}(t)\Big(Q_{i}^{(k)}(t)-Q_j^{(k)}(t)+\beta  \Delta_{ij}(t)\Big) + 2\sum_{i,k}\lambda_i^{(k)}Q_i^{(k)}(t)\\
& = \beta c_1 + c_2 - 2\sum_{i,j}\sum_k   f_{ij}^{(k)} \Big(Q_{i}^{(k)}(t)-Q_j^{(k)}(t)+\beta  \Delta_{ij}(t)\Big) + 2\sum_{i,k}\lambda_i^{(k)}Q_i^{(k)}(t)\\
& = \beta c_1 + c_2 + 2\sum_{i,k} Q_i^{(k)}(t)\Big(\lambda_i^{(k)} + \sum_{j} f_{ji}^{(k)} - \sum_{j} f_{ij}^{(k)}\Big) - \beta \sum_{ij}   \sum_k \Delta_{ij}(t)\Big( f_{ij}^{(k)}- f_{ji}^{(k)}\Big)\\
& \le \beta c_1 + c_2,
\end{split}
\]
where $c_1$ and $c_2$ are some constants, the second inequality is due to the operation of \routingname (see \eqref{eq:maxweight}), and the last inequality is due to \eqref{eq:flow-conservation} and \eqref{eq:channel-balance}. Using the law of iterated expectations
yields:
\[
\mathbb{E}[\Phi(\tau+1)]-\mathbb{E}[\Phi(\tau)]\le \beta c_1 + c_2.
\]
Summing over $\tau=0,\cdots,t-1$, we have
\[
\mathbb{E}[\Phi(t)]-\mathbb{E}[\Phi(0)] \le(\beta c_1 + c_2) \cdot t.
\]
Then we have 
\begin{equation}\label{eq:drift-sum}
\begin{split}
\sum_{i,k}\mathbb{E}\Big[\Big(Q_i^{(k)}(t)\Big)^2\Big] +  \frac{\beta}{2}\sum_{i,j}  \mathbb{E}\Big[\Delta_{ij}^2(t)\Big] &\le (\beta c_1 + c_2)t + \mathbb{E}[\Phi(0)].
\end{split}
\end{equation}
To show that \routingname achieves channel balance, we note from \eqref{eq:drift-sum} that
\[
\frac{\beta}{2} \Big(\mathbb{E}[\sum_{i,j} |\Delta_{ij}(t)|]\Big)^2\le \frac{\beta}{2}\sum_{i,j}  \mathbb{E}\Big[\Delta_{ij}^2(t)\Big]\le (\beta c_1 + c_2) t + \mathbb{E}[\Phi(0)],
\]
where the first inequality holds because the variance of $|\bm{\Delta}(t)|$ cannot be negative, i.e., $\text{Var}(|\bm{\Delta}(t)|)=\mathbb{E}\Big[\sum_{i,j}\Delta_{ij}^2(t)\Big] -  \Big(\mathbb{E}[\sum_{i,j} |\Delta_{ij}(t)|]\Big)^2 \ge 0$. Thus we have
\[
\mathbb{E}[\sum_{i,j} |\Delta_{ij}(t)|]\le \sqrt{2c_1t +\frac{2c_2 t}{\beta} +\frac{2 \mathbb{E}[\Phi(0)]}{\beta}}.
\]
Since $\mathbb{E}[\Phi(0)]<\infty$, we have that for any payment channel $i\leftrightarrow j$
\[
\lim_{t\rightarrow\infty} \mathbb{E}\Big[ \frac{|\Delta_{ij}(t)|}{t} \Big] = 0,
\]
i.e., the network maintains channel balance under the \routingname algorithm.

Similarly, we can show that
\[
\sum_{i,k}\mathbb{E}[Q_{i}^{(k)}(t)] \le \sqrt{(\beta c_1 +c_2)t + \mathbb{E}[\Phi(0)]},
\]
which implies that
\[
\lim_{t\rightarrow\infty}\frac{\mathbb{E}[Q_i^{(k)}(t)]}{t} = 0,~\forall i,k,
\]
i.e., the network is stable under the \routingname algorithm.

\subsection{Discussions of \routingname}
\subsubsection{Failure Resilience}
Due to its adaptive and multi-path nature, the \routingname algorithm is inherently robust against network failures. For example, when there are unresponsive nodes, \routingname  can quickly adapt and support the maximum possible throughput over the remaining available nodes.

\subsubsection{Privacy}
Due to the multi-path nature of \routingname, any intermediate node can only access the information for a small fraction of each value transfer request. As a result, the \routingname algorithm naturally provides good privacy protection in terms of the amount of transferred values. However, in \routingname, each node does need to know the destination of each value transfer request in order to place it in a proper debt queue. If we also need to hide the payment destination, onion routing \cite{onion} can be used in conjunction with \routingname. In onion routing, messages are encapsulated in layers of encryption. The encrypted data is transmitted through a series of network nodes called \emph{onion nodes}, each of which ``peels" away a single layer, uncovering the data's next destination. When the final layer is decrypted, the message arrives at its destination. We can direct payments through an overlay network consisting of onion nodes and apply \routingname to optimally route payments among onion nodes.

\subsection{Simulation Results}
\label{sec:routing:eval}
We numerically compare the performance of \routingname with two existing routing algorithms in payment state channel networks: SpeedyMurmurs \cite{speedy} and Flare \cite{flare} (used in Lightning Network). The simulation is conducted on the topology given in Figure \ref{fig:topology-medium}.

\begin{figure}[ht!]
\begin{center}
\includegraphics[width=2.8in]{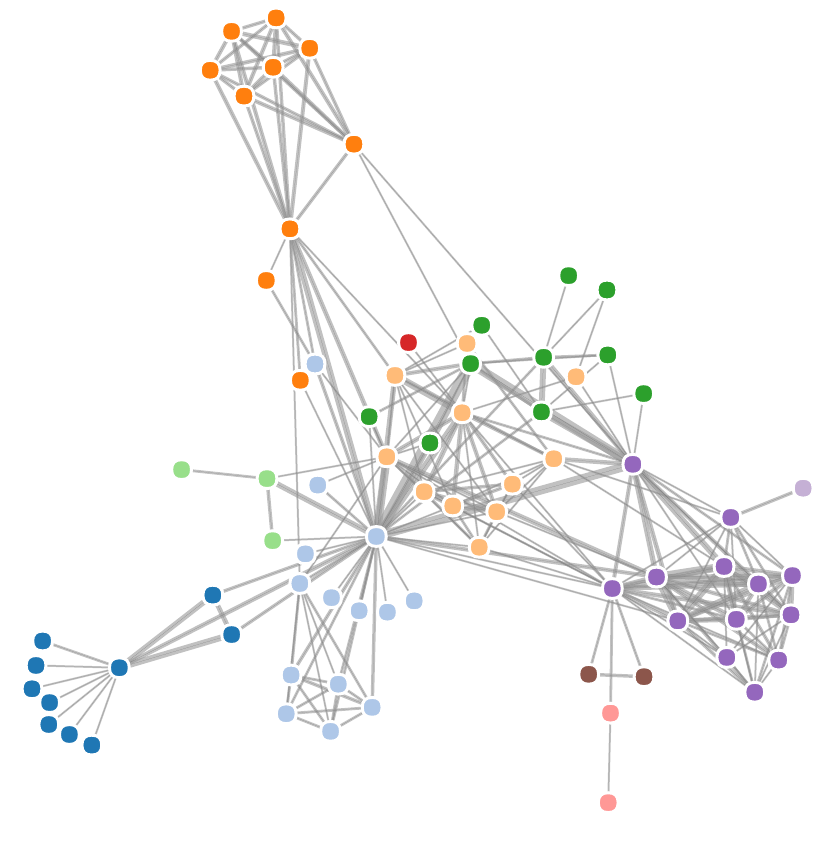}
\caption{Payment channel network topology used in simulations (77 nodes, 254 bi-directional payment channels). The payment channel network is running with 40 payment flows with randomly chosen source-destination pairs. The initial deposit for each channel is uniformly distributed within [100,200] tokens. Payment arrivals follow a Poisson process and the size of each payment follows a geometric distribution with the mean of 3 tokens.}
\label{fig:topology-medium}
\end{center}
\end{figure}

\begin{figure*}[ht!]
\centering
\begin{minipage}[t]{0.45\textwidth}
\centering
\includegraphics[width=2.2in, height=1.9in]{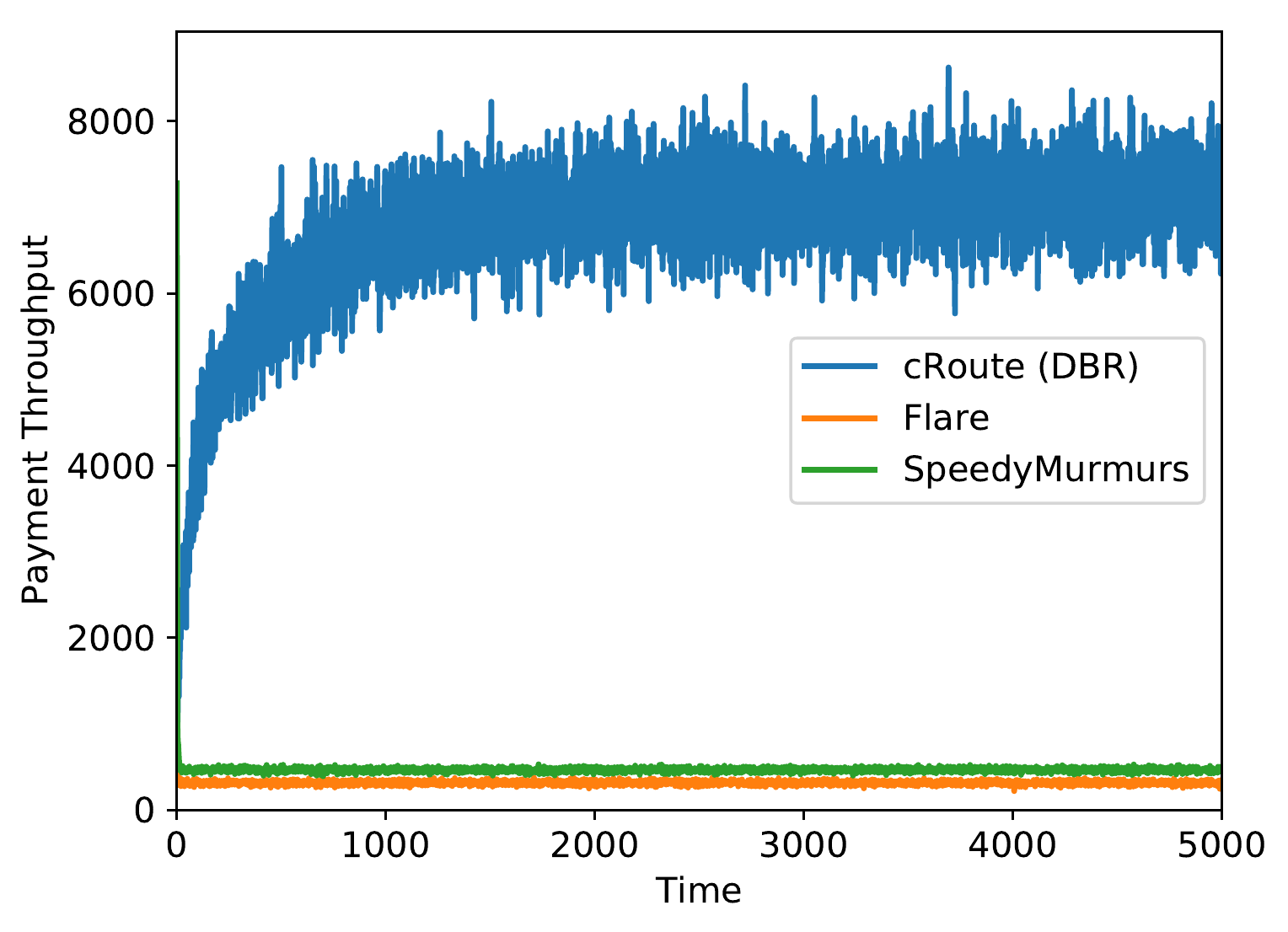}
\caption{Instant Payment throughput comparison among DBR (avg: 6748 payments/slot), SpeedyMurmurs (avg: 467 payments/slot) and Flare (avg: 316 payments/slot).}
\label{fig:instant_throughput-medium}
\end{minipage}%
\hspace{0.5cm}
\begin{minipage}[t]{0.45\textwidth}
\centering
\includegraphics[width=2.2in, height=1.9in]{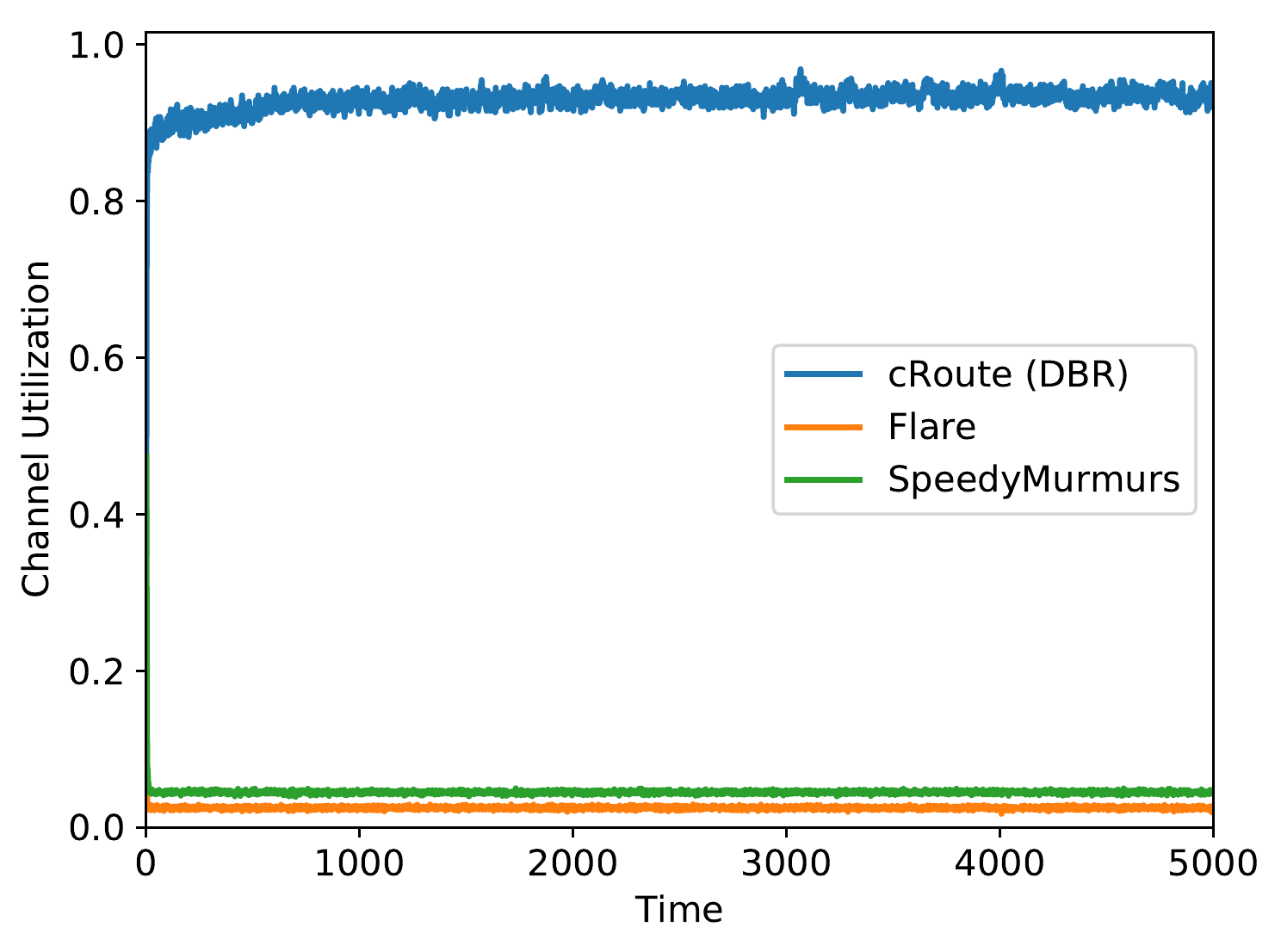}
\caption{Channel utilization comparison among DBR, SpeedyMurmurs and Flare. Higher channel utilization implies a higher level of channel balancing.}
\label{fig:channel-medium}
\end{minipage}
\end{figure*}



Figure \ref{fig:instant_throughput-medium} shows the throughput performance comparison. It is observed that \routingname achieves 15x improvement in average payment throughput as compared to the existing routing algorithms. The exceptional throughput performance of \routingname is due to its channel-balancing and congestion-aware nature. In particular, Figure \ref{fig:channel-medium} illustrates the channel utilization\footnote{Channel utilization corresponds to the ratio between the amount of transferred tokens in each time slot and the total token deposits of all channels. For example, if the total amount of deposits is 100 and only 50 of them are moved in slot $t$, then the overall channel utilization in that slot is 50\%.} under \routingname, SpeedyMurmurs and Flare, where we can observe that \routingname consistently achieves high (nearly 100\%) channel utilization while the other routing algorithms only achieve less than 5\% channel utilization due to the lack of channel balancing.
\section{cOS: Off-chain Decentralized Application Operating System}
\label{sec:cOS}


To help everyone quickly build, operate, and use scalable off-chain decentralized applications without being hassled by the additional complexities introduced by off-chain scaling, \systemname innovates on a higher level abstraction: cOS, a combination of application development framework (SDK) and runtime system. This section provides the high-level vision, design objectives and illustrations on cOS. 

\subsection{Directed Acyclic Graph of Conditionally Dependent States}
\label{sec:dag}
In this section, we provide a view of our abstraction model on the construct of off-chain applications, and describe how the model is integrated with state channel networks. In order to support use cases beyond simple P2P payments, we model a system of off-chain applications as a directed acyclic graph (DAG) of conditionally dependent states, where the edges represent the dependencies among them.

\begin{figure}[ht!]
\begin{center}
\includegraphics[width=4in]{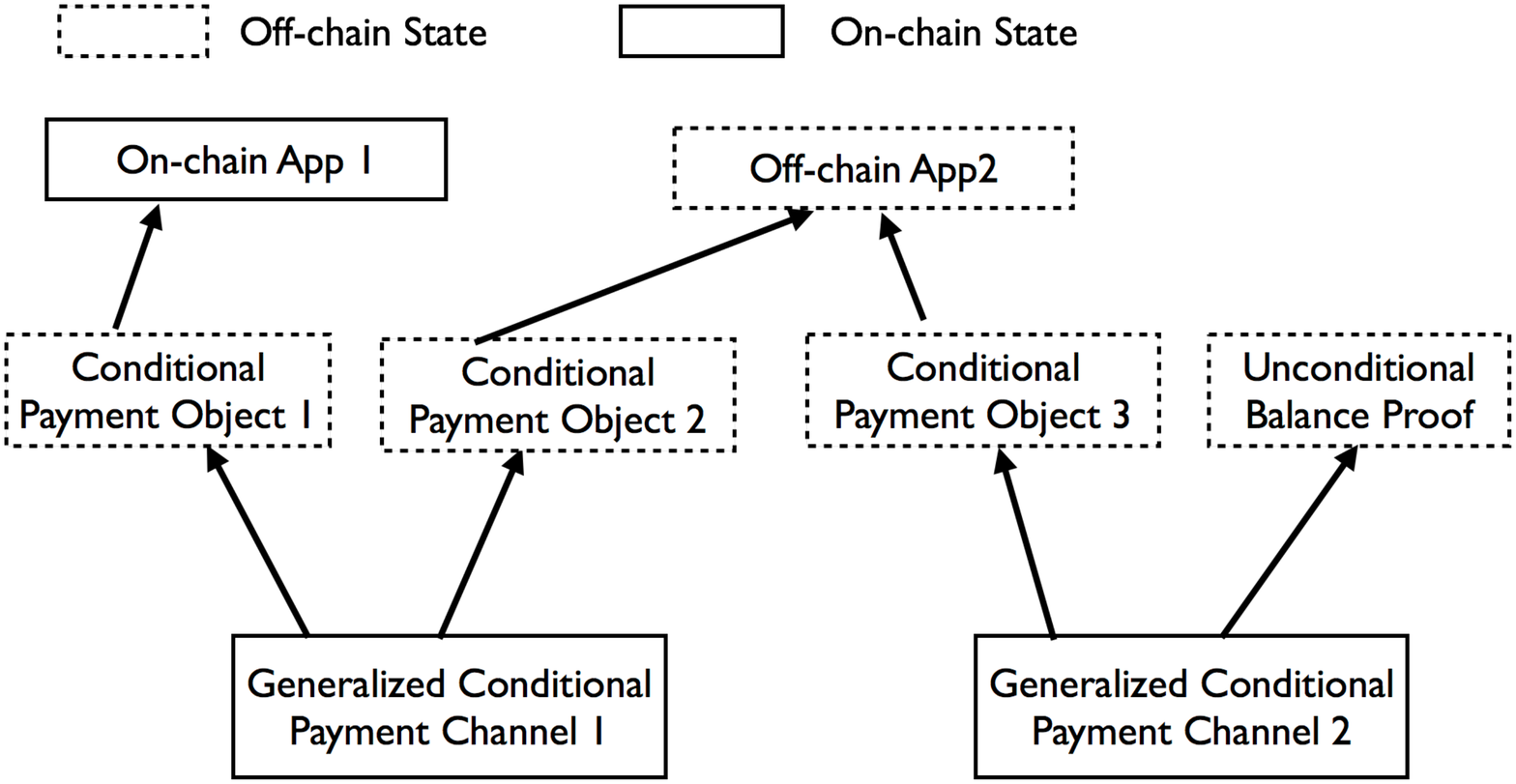}
\caption{DAG of Conditionally Dependent States}
\label{fig:off_chain_dag}
\end{center}
\end{figure}

Figure~\ref{fig:off_chain_dag} illustrates the system model, where the Generalized Conditional Payment Channel in the payment networks are the only contracts with on-chain states. The settlement of these on-chain states depends on one or more conditional payment objects (e.g. Conditional Payment Object 3), which are completely off-chain but on-chain enforceable. We want to highlight that these conditional payment objects are not only simple time hash locked transactions, but can be conditioned on off-chain application contract states, such as ``Off-chain App X'' in Figure~\ref{fig:off_chain_dag}. 

The conditional payment objects can be relayed through multiple hops just like simple unconditional payments objects. For example, Payment Channel 1 can be a channel connecting Alice and Bob and Payment Channel 2 can be a channel connecting Bob and Carl. Let ``Off-chain App 2'' be an off-chain chess game Alice is playing with Carl, and suppose Alice wants to express the semantic of ``Alice will pay Carl 10 ETH if Carl wins the game''. Even without a direct channel between Alice and Carl, Alice can send a conditional payment to Carl through Bob with two layers of conditional locks. The first layer is a simple time hashed lock to make sure that Bob relays and resolves the payment in a reasonable amount of time. The second layer locks the payment conditioning on the result of the chess game. With this two-hop relay, the conditional payment between Alice and Carl can be settled via Bob even though Bob did not involve in the chess game. This is a minimized example of how a dependency graph formed by generalized state channels, conditional payment objects and off-chain applications can support arbitrarily complex multi-party interactions.

Note that off-chain objects do not have to depend only on off-chain objects. For example, Alice can pay Carl when the latter successfully transfers a certain ENS name to the former. In other words, the payment depends on the off-chain condition that the owner of the ENS name changes from Carl to Alice.

Also, off-chain payment objects do not always have to be conditional: a conditional payment object can ``degenerate'' into an unconditional balance proof as the application runs. More generally speaking, conditional dependencies are transient by nature: application state updates are done via a pair of two topological traversals of the underlying state graph. The first traversal goes in the forward direction and the second one in the reverse direction. The forward traversal, starting from the on-chain state-channel contracts, creates additional transient conditional dependency edges and modifies existing ones. The reverse traversal may remove existing transient conditional dependency edges, because some conditions evaluate to constant true when traversing backward.


\subsection{Off-chain Application Development Framework}

Much like how modern high-level languages and operating systems abstract away the details about the underlying hardware, the complexity of interacting with conditional state dependency graphs necessitates a dedicated development framework. With the principle of ``ease of use'' on our mind, \systemname presents the cOS SDK, a complete toolchain solution for the creation, tracking, and resolution of states in off-chain applications. We hope that the SDK will accelerate the adoption of the off-chain scaling solution and the payment network provided by \systemname, fostering a strong ecosystem.

\begin{figure}[ht!]
\begin{center}
\includegraphics[width=4in]{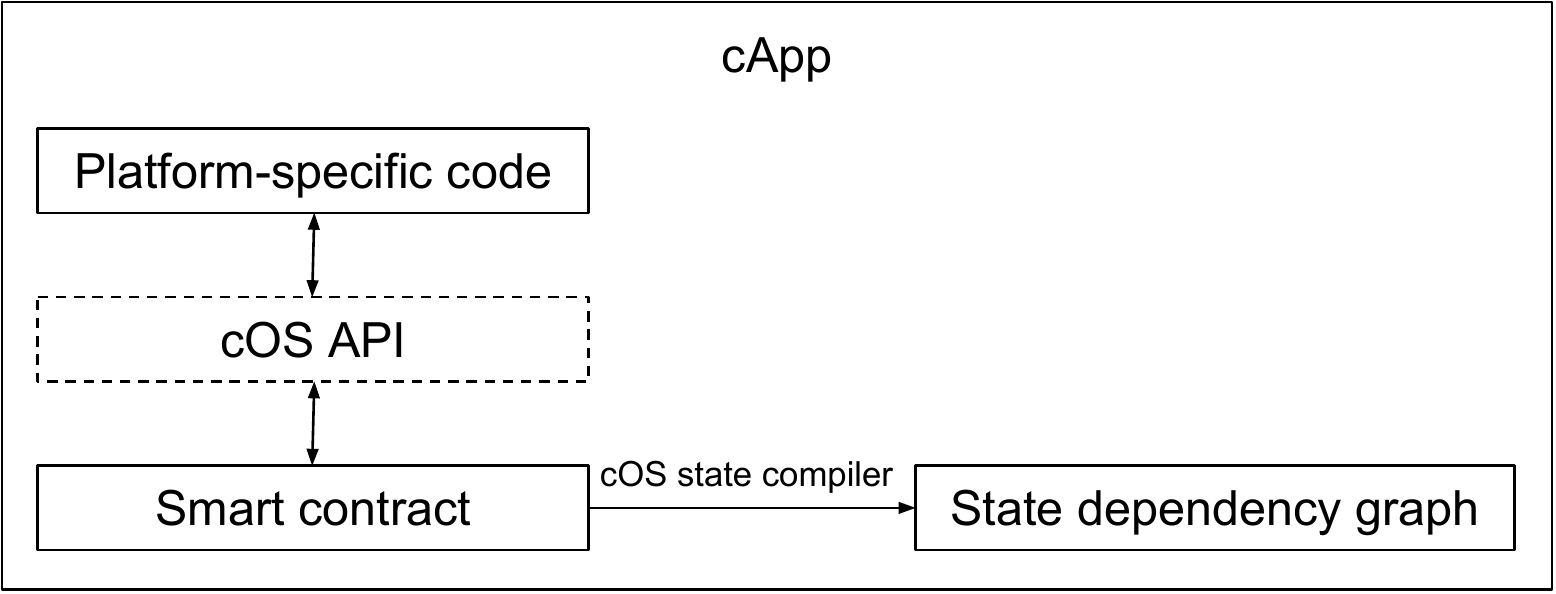}
\caption{Structure of a decentralized application on Celer Network (cApp)}
\label{fig:cApp_structure}
\end{center}
\end{figure}

In general, we categorize decentralized applications into two classes: simple pay-per-use applications and more complex multi-party applications. The pay-per-use applications include examples like the Orchid Protocol, where the user keeps receiving microservices (e.g. data relay) from a real-world entity and streams payments through the payment network. Since there is no need for conditional dependency on other off-chain states, a lean transport layer API on top of the routing layer, both of which are provided by \systemname, suffice for such cases.

The class of multi-party applications, the general structure of which is illustrated in Figure~\ref{fig:cApp_structure}, is where the idea of conditional state dependency graphs really shines. The SDK defines a set of design patterns and a common framework for developers to express the conditional dependencies. We plan to extend the existing smart contract languages with modern software construction techniques such as metaprogramming, annotation processing, and dependency injection so that the dependency information can be written out explicitly without being too intrusive. A compiler then processes the application code, extracts the declared off-chain objects, and generates the conditional dependency graph. The compiler detects invalid or unfulfillable dependency information and generates human-readable errors to assist the developer in debugging. To help developers reason about the dependencies even further, the SDK will be able to serialize the graphs into common formats such as Graphviz, with which they can be easily visualized and presented.

The SDK also provides a code generator that generates a set of ``bridge methods'' for interacting with smart contracts whose code is available at compile time. The code generator parses the application binary interface (ABI), which specifies the signature of all callable functions in a smart contract, and generates the corresponding bridge methods in platform-specific languages such as Java. The main advantage of this approach is type safety: the glue methods replicate the method signatures of the functions in the smart contract faithfully, providing a static and robust compile-time check before dispatching the method to the cOS runtime for execution.

\subsection{Off-chain Application Runtime}

The cOS runtime serves as the interface between cApps\footnote{We name decentralized applications running on Celer Network as cApps.} and the \systemname transport layer. It supports cApps in terms of both network communication and local off-chain state management. The overall architecture is illustrated in Figure~\ref{fig:cOS_architecture}.

\begin{figure}[ht!]
\begin{center}
\includegraphics[width=4in]{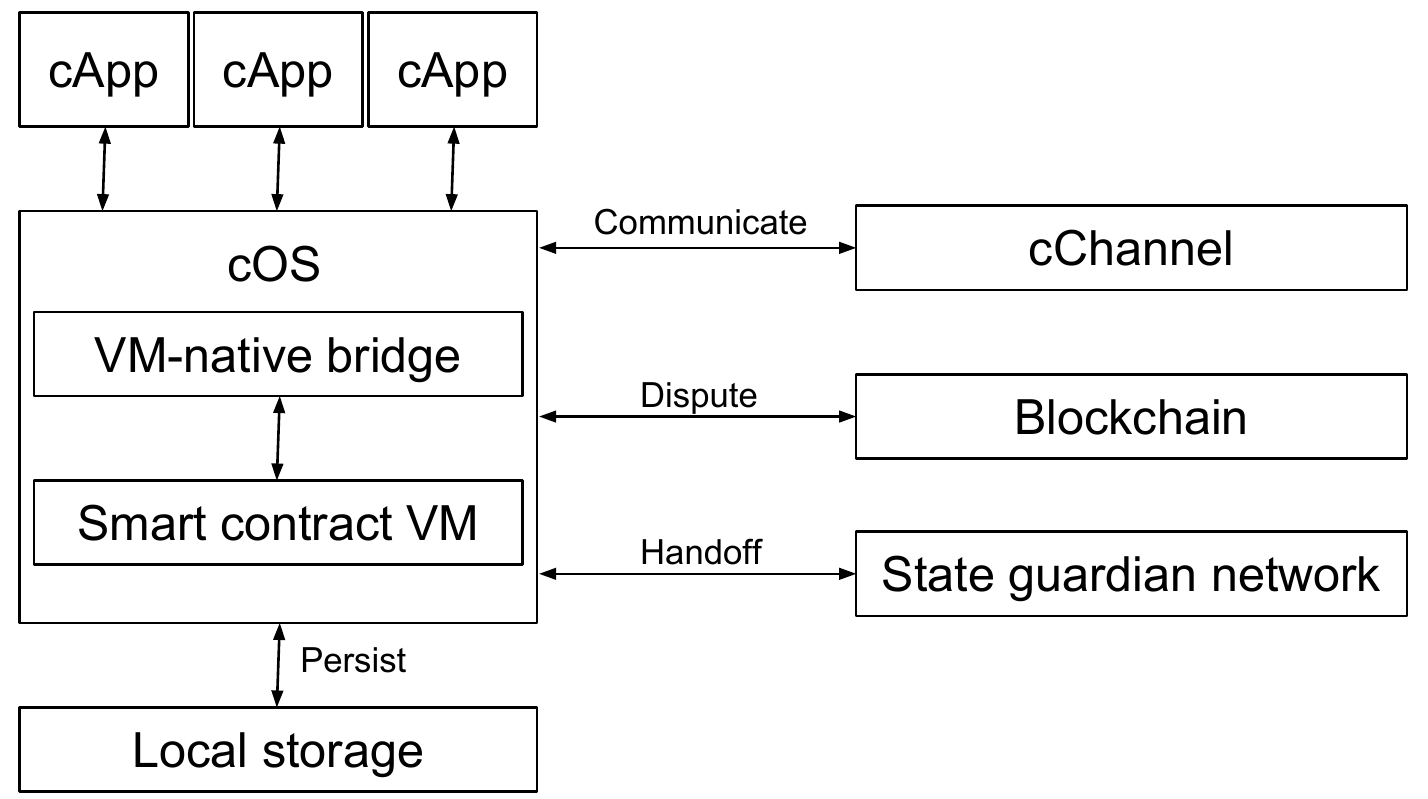}
\caption{cOS Runtime Architecture}
\label{fig:cOS_architecture}
\end{center}
\end{figure}

On the network front, the runtime handles multi-party communication during the lifecycle of a cApp. It also provides a set of primitives for secure multi-party computation capable of supporting complex use cases such as gaming. In the case of counter-party failure, whether fail-stop or Byzantine, the runtime relays disputes to the on-chain state. In the case of the client going offline, the runtime handles availability offloading to the State Guardian Network. When the client comes back online, the runtime synchronizes the local states with the State Guardian Network.

For local off-chain state management, the conditional state graphs synthesized by the cOS SDK is bundled within the cApp and passed on to the runtime for off-chain execution. The runtime serves as the infrastructure to create, update, store and monitor off-chain states locally on \systemname clients. It tracks the internal logic of the applications running on top of it and performs the DAG traversal of state updates as outlined in \S~\ref{sec:dag}. It also gracefully handles payment reliability issues such as insufficient capacity for routing the payment. 
 
At its core, the cOS runtime bundles a native virtual machine (VM) for running smart contracts. While we intend to deploy cOS to many platforms including desktop, web, mobile and IoT devices, we have adopted the ambitious design principle of ``write once, run anywhere''. In other words, we enable developers to write the common business logic once and run the exact same on-chain smart contract code in every environment as opposed to having to implement multiple variants of the same logic. By adopting this principle, we aim to eliminate code duplication and ensure high degree of consistency across various platforms.

The platform-specific part of cApps, such as user interface (UI), can be built in languages most suitable to each platform (eg. Kotlin for Android and Swift for iOS). The UI code is also free to use platform-specific utilities and libraries, so that the look and feel of cApps match the respective design guidelines on each platform.

The cOS runtime provides VM-native bridge implementations in different languages for the platform-specific code to interact with the underlying business logic. For example, consider a cApp representing a chess game running on iOS with the UI written in Swift and business logic written in Solidity. Naturally, the UI layer will need to query the cOS VM for the state of the game board, and it will be able to do so via the Solidity-Swift bridge. Because the code for the contract is available at compile time, the code generator cOS SDK would have generated a bridge method named $chess.getBoardState$, which is dispatched to the VM for the actual query. Whenever possible, we make use of the language's foreign function interface (eg. JNI) to reduce the overhead of calling back-and-forth between smart contracts and the native code. The developer will also be able to use the same debugging and profiling tools for on-chain smart contracts in the off-chain development scenario.

In order to genuinely replicate the state changes that would have happened on-chain in the off-chain environment, the VM progresses through the same bytecode as if they were executed on-chain, with the caveat of a few differences. The first major difference is that the VM needs to update and store the states locally instead of on the blockchain. To achieve seamless and transparent inter-operation between the VM and the rest of cOS, we will implement a set of APIs that bridge platform-specific storage backends with the VM. The second major difference is that as opposed to being always online, a local VM can shut down unexpectedly at any time due to software bugs, hardware failure or simply loss of power. To avoid corruption of local states, we need to implement a robust logging, checkpointing and committing protocol. A third minor difference is that the logic for gas metering can be omitted, because the execution happens locally and it does not make sense to charge gas fees.

The bundled VM needs to be lightweight and performant so that it can run well on mobile and IoT devices, which tend to operate under tight processor power, memory capacity and battery life constraints. While we currently embed a lightweight Ethereum VM in cOS, we are researching into adopting more common bytecode formats (eg. WebAssembly) with the goal of supporting more contract languages and other blockchains.

In our ultimate vision of the cOS VM, we will apply modern VM techniques such as ahead-of-time compilation (AOT) and just-in-time compilation (JIT) to achieve near-native performance of off-chain smart contract execution. Instead of interpreting the smart contract bytecodes like what most Ethereum VMs currently do, we compile the bytecodes to lower level intermediate representations that are closer to native code. If the code for a certain contract is available at compile time (eg. a contract that is already deployed on-chain), we perform the compilation ahead of time and statically link the binary with the rest of the application. For the contracts that are dynamically loaded at runtime, we profile them for frequently-called functions (i.e. ``hot'' code) and perform just-in-time compilation. We believe that the combination of these two techniques will bring a great balance between performance and energy consumption, which are both crucial for mobile and IoT devices.

\section{cEconomy: Off-chain Cryptoeconomics Mechanism Design}
\label{sec:econ}

The native digital cryptographically-secured protocol token of the Celer Network, (\tokenname) is a major component of the ecosystem on the Celer Network, and is designed to be used solely on the network. \tokenname is a non-refundable functional utility token which will be used as the platform currency in the ecosystem on the Celer Network. \tokenname does not in any way represent any shareholding, participation, right, title, or interest in the Token Vendor, the Foundation, their affiliates, or any other company, enterprise or undertaking, nor will \tokenname entitle token holders to any promise of fees, revenue, profits or investment returns, and are not intended to constitute securities in Singapore or any relevant jurisdiction. \tokenname may only be utilized on the Celer Network, and ownership of \tokenname carries no rights, express or implied, other than the right to use \tokenname as a means to enable usage of and interaction with the Celer Network.

In the following, we introduce Celer Network's cryptoeconomics mechanisms, cEconomy, whose design is based on the principle that a good cryptoeconomics model (token model) should provide additional values and enable new game-theoretical dynamics that are otherwise impossible. 
In the following, we first elaborate the fundamental tradeoffs in off-chain ecosystems (Section \ref{sec:tradeoff}) and then demonstrate how cEconomy can bring value and enable new dynamics to ``balance out" those tradeoffs (Section \ref{sec:ceconomy_design}).

\subsection{Tradeoffs in Off-chain Ecosystems}\label{sec:tradeoff}
Any off-chain solution, while gaining scalability, is also making tradeoffs. In the following, we describe two fundamental tradeoffs in off-chain ecosystems: scalability-liquidity tradeoffs and  scalability-availability tradeoffs.
\subsubsection{Off-Chain Scalability vs. Liquidity}
Off-chain platform gains scalability by first trading off network liquidity. For example, in a bi-party payment state channel, the two involved parties can safely send each other payments at high speeds without hitting the underlying blockchain because they have deposited liquidity to the on-chain bond contract at the beginning. Liquidity-locking of this nature works is fine for the end users because the end users can simply deposit their own liquidity to the open channels and enjoy the scalable dApps. However, it poses a significant challenge for those who want to operate as an Off-chain Service Providers (OSPs). Using state channels as an example,  OSPs need to make deposits in each channel with outgoing payment possibility. Those deposits can easily aggregate to an astronomical amount. Even though Celer Network’s sidechain channels can significantly reduce the level of liquidity requirement, each block proposer still needs to deposit fraud-proof bonds proportional to the level of value transfer ``at stake". 

All in all, significant amount of liquidity is needed to provide effective off-chain services for global blockchain users. However, whales may not have the business interest or technical capability to run an off-chain service infrastructure, while people who have the technical capability of running a reliable and scalable off-chain service often do not have enough capital for channel deposits or fraud-proof bonds. Such a mismatch creates a huge hurdle for the mass adoption and technical evolution of off-chain platforms. If not mitigated, eventually only the rich can serve as OSPs. This high capital barrier of becoming an OSP will result in a centralized network that providing undermines the entire premise of blockchain's decentralization vision. From a more practical view, censorship, poor service quality and privacy breach will hurt users as today's centralized services do. 

\subsubsection{Off-Chain Scalability vs. Availability}\label{sec:availability}
While an off-chain platform improves scalability by bringing application states off-chain, it imposes an impractical ``always online" responsibility on the users, because the off-chain states should always be available for on-chain disputes. For example, in a bi-party payment state channel, if one party goes offline, the counterparty may get hacked or act maliciously, and try to settle an old but more favorable state for himself. The data availability issue is even more critical in a sidechain channel where block proposers need to be independently monitored and validated while the participants are offline; this is a matter of security and should be scrutinized carefully. This challenge is even more critical in machine to machine communication scenarios where IoT devices are not likely to be online all the time. Therefore, it is crucial to design proper mechanisms that guarantee data availability in an off-chain platform. Solving this challenge requires systematic thinking of the entire off-chain ecosystem and existing solutions all fail to provide the important properties of decentralization, efficiency, simplicity, flexibility, and security as we will discuss more in the following section.

\subsection{cEconomy Design}\label{sec:ceconomy_design}
To balance the above-mentioned tradeoffs, we propose a suite of cryptoeconomics mechanisms called cEconomy that includes three tightly interconnected components: Proof of Liquidity Commitment (PoLC) mining, Liquidity Backing Auction (LiBA) and State Guardian Network (SGN).
The relationship among the three components is illustrated in Figure \ref{fig:economy}.

\begin{figure}[t]
\begin{center}
\includegraphics[width=5.5in]{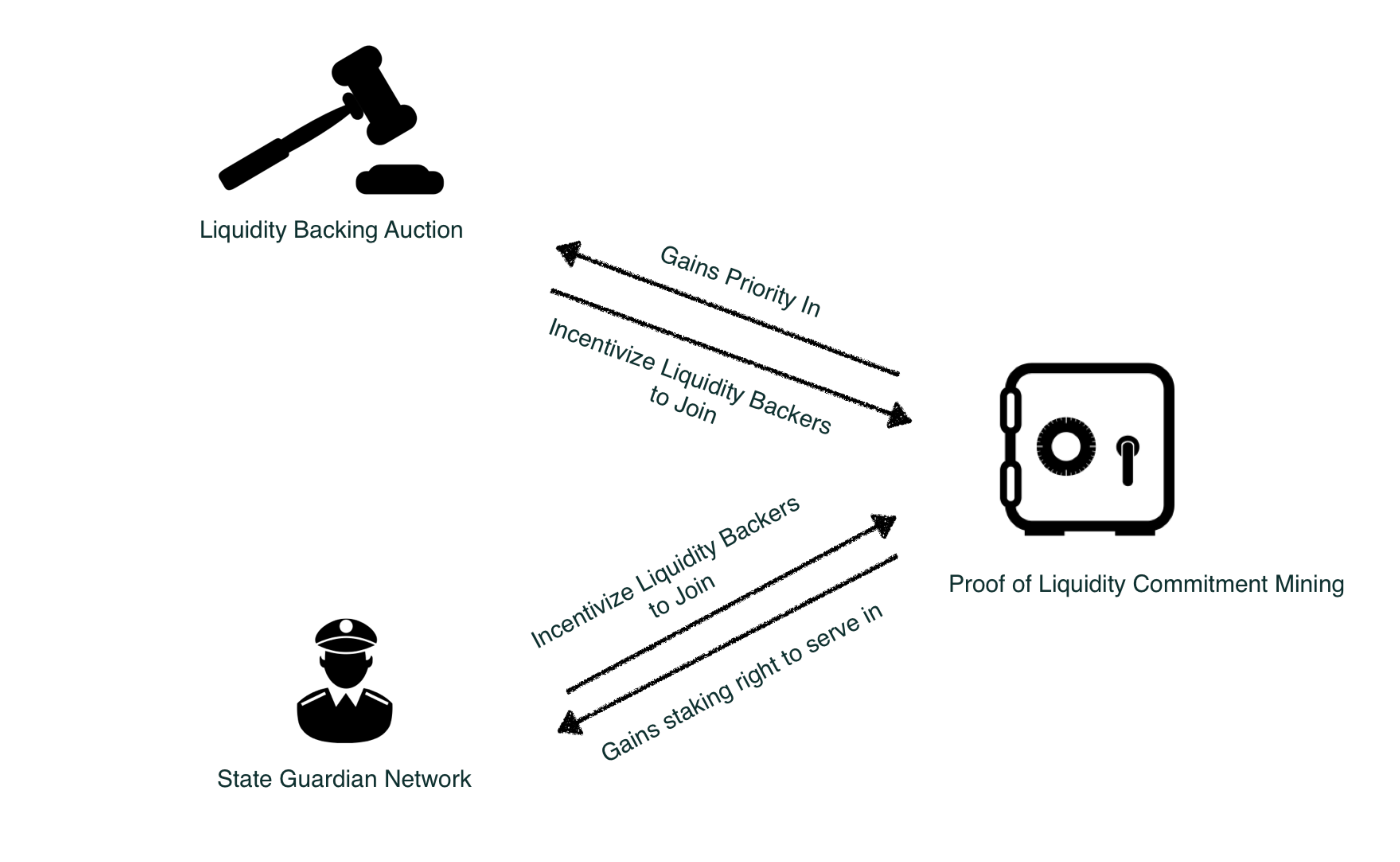}
\caption{Relationship among cEconomy components.}
\label{fig:economy}
\end{center}
\end{figure}

Before moving on to the details of these components, we first introduce several terms that will be used throughout this section. Specifically, a user in our cEconomy system may play any of the three roles: Off-chain Service Provider (OSP), End Users (EU), Network Liquidity Backer (NLB) and State Guardians (SG). Off-chain Service Providers (OSP) are entities who have the technical capability to run highly redundant, scalable and secure off-chain infrastructures. End Users (EU) can access the off-chain services provided by OSP (e.g., pay and receive cryptocurrency). They can be common consumers or they can be IoT devices, VPN providers, live video streaming providers and CDN providers, counterparties in Machine to Machine (M2M) systems, or even an off-chain/on-chain smart contract. Network Liquidity Backers (NLB) are entities that lock up their liquidity in the system to support the operations of off-chain infrastructure. State Guardians are those who provide EUs decentralized, secure, flexible and efficient state guarding service through State Guardian Network.

\subsubsection{Proof of Liquidity Commitment (PoLC) Mining}
Our first goal is to balance out the scalability-liquidity tradeoff by lowering the liquidity barrier for technically capable parties to become off-chain service providers and thus creating an efficient and competitive market for good and reliable off-chain services. The gist of the idea is to enable service providers to tap into large amounts of liquidity whenever they need to. The first part to realize this idea is to provision an abundant and stable liquidity pool that can smooth out short-term liquidity supply fluctuation. To that end, we propose the Proof of Liquidity Commitment (PoLC) virtual mining process.

From a high level, the PoLC mining process is to incentivize Network Liquidity Backers (NLB) to lock in their liquidity (which can be in the form of digital assets, including but not limited to cryptocurrencies and CELR) in \systemname for a long time by rewarding them with \tokenname tokens and therefore establishing a stable and abundant liquidity pool. 


More specifically, the mining process involves NLBs to commit (lock) their idle liquidity (for example, ETH) to a ``dumb box", called Collateral Commitment Contract (CCC), for a certain period of time. During this period of time when the digital assets are locked, the NLB's assets can only be used in the liquidity backing process and nothing else. More formally, the PoLC mining process can be defined as the following.

\begin{definition}[\textbf{PoLC Power}]
 If NLB $i$ locks $S_i$ amount of local cryptocurrency in a blockchain (e.g. ETH) for $T_i$ time, its PoLC power $M_i$ is computed as 
\begin{equation}
M_i = S_i \times T_i.
\end{equation}
\end{definition}

\begin{definition}[\textbf{PoLC Incentive Mechanism}]
 For a limited period of time, \systemname intends to provide incentives in the form of \tokenname to NLBs who lock their CCC as a show of support for the system. Incentives will be distributed proportional to each NLB's PoLC power. Let $R_i$ denote the incentives of $i$, one has:
\begin{equation}
R_i = \frac{R \times M_i}{\sum\limits_{j=1}^N M_j},
\end{equation}
where $R$ is the total reward for the current block.
\end{definition}

Note that locking liquidity in CCC does not carry any inherent counterparty risk as it simply shows a liquidity commitment to Celer Network. Also, note that early unlocking of CCC is not allowed. One may try to create a ``spoofed liquidation'' with an appearance of one's CCC getting liquefied due to ``hacking'' of a faked OSP. To prevent this spoofing, the newly mined \tokenname is not available for withdrawal and usage until CCC unlocks. Any early liquidation will cause the already mined CCC to be forfeited and redistributed to other miners. The construct of a common denominator of liquidity in PoLC is also an important question. For the initial launch of the platform, we will use the native currency of the target blockchain and later use more heterogeneous crypto assets through external price oracles.  


With these mechanisms in place, the PoLC mining process ensures that the PoLC power in the system will grow as the system and utility of the \tokenname grows, forming a positive loop.

At this point, one may wonder why is \tokenname valuable so that it can act as such an incentive? We explain that in the following sections describing Liquidity Backing Auction and State Guardian Networks.

\subsubsection{Liquidity Backing Auction (LiBA)}
The second part for solving the liquidity puzzle is to enable a way for off-chain service providers to access to liquidity pool globally, which is achieved via the Liquidity Backing Auction (LiBA). LiBA enables off-chain service providers to solicit liquidity through ``crowd lending". In essence, an off-chain service provider starts a LiBA on Celer Network to ``borrow" a certain amount of liquidity for a certain amount of time. An interested liquidity backer can submit a bid that contains the interest rate to be offered, amount of liquidity and the amount of \tokenname that she is willing to stake for the said period of time. 
The amount of liquidity can be submitted via a CCC. That is, CCC has the functionality to act as a liquidity backing asset. The borrowed liquidity will be used as a fraud-proof bond or outgoing channel deposit.

LiBA is a generalized multi-attribute Vickrey-Clarke-Groves (“sealed-bid second-score”) auction. To start an auction process, an OSP creates a standard LiBA contract through the \systemname's central LiBA registry with information regarding the total amount of requested liquidity ($q$), duration of the request ($d$) and the highest interest rate ($r_{max}$) that it can accept. NLBs who watch the registry will notice this new LiBA contract and can start the bidding process. \systemname requires all crypto assets to be locked in CCC for the bidding process. Note that CCC can be ``lock-free'' and simply used as a backing asset without the functionality of PoLC mining. CCC acts as a container for crypto assets and provides a unified verifiable value of heterogeneous crypto assets. Moreover, the use of CCC makes it easier for NLBs to participate in LiBA without moving crypto assets around every time they bid and thus simplifies the backing process and improve security. NLB $i$ submits the bid in the form of a tuple $b_i=(r_i, t_i, c_i)$, where $r_i$ is interest rate, $t_i$ is the total amount of \tokenname it is willing to lock up during the contract time and $c_i$ is the aggregate currency value contained in the set of CCCs bonded with this bid. Once the bid is submitted, the corresponding CCCs are temporarily frozen. After sealed bidding, the LiBA contract uses reverse second-score auction~\cite{auction} to determine winning bids with the following three steps. 
\begin{itemize} [leftmargin=*]
\item{(\textbf{Scoring Rule}). For each bid $b_i=(r_i, t_i, c_i)$ in the bid set $\mathcal{B}=\{b_1, b_2, ..., b_n\}$ with $f_i = \frac{t_i}{c_i}$, its score $s_i$ is calculated as the following:
\begin{equation}
s(b_i) = w_1\frac{f_i}{f_{max}} - w_2\frac{r_i}{r_{max}},
\end{equation}
where $f_{max} = \max\{f_1, f_2, ..., f_n\}$ and $r_{max} = \max\{r_1, r_2, ..., r_n\}$. $w_1$ and $w_2$ are weights for the two components and are initially to ensure we take into account an interest rate with higher weight and then take into account the amount of staked \tokenname. \footnote{These weights are subject to future decentralized governance adjustment.}.}
\item{(\textbf{Winner Determination}). To determine who has the opportunity to become the network liquidity backer, the LiBA contract sorts the bids in $\mathcal{B}$ in descending order by their scores. The sorted bid set is denoted by $\mathcal{B}^* = \{b^*_1, b^*_2, b^*_3, ..., b^*_n\}$, where $s(b^*_1) \ge  s(b^*_2)\ge\cdots \ge s(b^*_n)$ (ties are broken randomly). Winners are the first $K$ bids in $\mathcal{B}^*$, where $\sum\limits_{i=1}^K t_i \geq q$ and $\sum\limits_{i=1}^{K-1} t_i < q$.}
\item{(\textbf{Second-Score \tokenname Staking/Consumption}). After winners are determined, their CCCs will be locked in the LiBA contract for time $d$ (the duration of the request), their interest requests are accepted and interests are prepaid by the OSP initiating the liquidity request. However, it is important to note that not all of their committed \tokenname are locked/consumed in this contract. Each winner only needs to lock up/consume enough \tokenname so its score matches the score of the first loser in this auction. 
Whether the token will be locked or consumed depends on the stage of the platform. In the first five years, new tokens will be generated through PoLC mining and LiBA only requires token staking. When the PoLC mining concludes, LiBA  will start to consume token and the consumed tokens will be injected into the system as continuous PoLC mining rewards. 
Note that under the second-score \tokenname staking/consumption mechanism, the participants are projected to submit bids matching their true valuation (truthfulness \cite{vcg}) of the good (in this case, the opportunity to back the network liquidity). }
\end{itemize}

\begin{framed}
\textbf{Example:} Assume that an OSP initiated a LiBA with the following parameters (600 ETH, 30 days, 1\%) and there are three potential bidders (let's say A, B, and C) for this LiBA. The three bidders' bids are $b_A$ = (1\%, 800 \tokenname, 400 ETH); $b_B$ = (0.5\%, 800 \tokenname, 200 ETH); $b_C$ = (1\%, 100 \tokenname, 400 ETH). According to the scoring rule, we have $s_B > s_A > s_C$. Since A and B can fill the entire request, they are selected as winners. It should be noted that even though A and C have the same interest rate (1\%) and provide the same amount of liquidity (400 ETH), bidder A is selected as a winner while bidder C loses; this is due to the fact that their committed \tokenname tokens, as a symbol of their contributions to this platform, are significantly different.  Finally, according to the second-score staking rule, A and B lock (or consume) their \tokenname tokens to match the score of C for 30 days.    
\end{framed}

After the auction process finishes, the OSP who initiated the liquidity request pays the interests to the wining liquidity backers by depositing into the LiBA contract. Upon receiving the payment of interests, the LiBA contract then gives the interests to the corresponding liquidity backers and issues 1:1 backed cETHs (using ETH as an example) that match the liquidity request amount. Although cETH is essentially an IOU, it brings no risk to the user as these IOUs are $100\%$ insured by the network liquidity backers in the LiBA contract. 

In normal cases, the LiBA contract is resolved before the timeout when the OSP sends back all the cETH tokens. Basically, before the timeout, the OSP will settle all paid cETHs to EUs with real ETHs by withdrawing from upstream channels collectively. 

In the case where the OSP may get hacked, \systemname's trust model can vary. The simplest trust model without any protocol-level overhead is reputation-driven, where NLBs choose a reputable OSP without any history of default. In this simple model, NLBs are exposed to the risk of losing funds and assets as their CCCs are insurances for the EUs if the OSP defaults. However, it is arguable that even in this simple trust model, operating a highly reliable and reputable OSP is possible; it is very unlikely that all backings will be lost. There are additional security features which may be added around LiBA to further alleviate the potential risk. For example, newly issued cETHs are only allowed to be deposited to a whitelist of state channel contracts; cETHs are only allowed to be used incrementally with an upper bound spending speed. There are also a lot of things an OSP can do to maintain a secure infrastructure such as compartmentalized multi-node deployment, formal verification of security access rule of network infrastructure and more. 

In addition, we enable an enhanced security model where a randomly selected quorum of NLB will need to co-sign an OSP's operations (e.g. payment). These NLBs will only allow an outgoing transfer if and only if they see an incoming transaction with matching amount. These NLBs are also tethered to the incoming payments of OSP. If OSP fails to make the repayment eventually, NLBs will have the first-priority right to claim the incoming funds to OSP from other channels. However, we do note that this operation model will inevitably tradeoff some efficiency of the network. 

Having said these, we believe the ultimate balance in the trust model should be defined by the market demand. We open both trust model for the market to organically evolve. We envision that the trust-free model will be more favorable in the early days of network launch and then it will become more trust-based. 

Regardless of the LiBA's trust model, we want to highlight that the LiBA process ensures that \textbf{end users never take any security risk} as the required liquidity is 100\% ``insured" by the LiBA contract. In Celer Network's system, we strive to make sure that the benevolent end users do not need to worry about the security of their received fund and LiBA achieves that. PoLC and LiBA together incentivize an abundant liquidity pool, lower the barrier of becoming an off-chain service provider, reduce centralization risk, and accelerate network adoption.

\subsubsection{State Guardian Network}
Another usage of \tokenname is to provide off-chain data availability with novel insurance model and simple interactions, which balances out the scalability-availability tradeoffs as mentioned in Section \ref{sec:availability}.

From the surface, the availability problem seems to be an easy one to solve. One possible answer to that question might be: let's build some monitoring services in the future and people will pay for these monitoring services when themselves are not online. It feels like a reasonable solution at first look, but we drive this train of thought just a little bit forward, we will immediately see track-wrecking flaws.

Let's start with this question: are these monitoring services trust-based? If the answer is yes, then it creates another centralized choking point, single point of failure and is just not secure. Malicious counterparty can easily bribe these monitoring services to hurt benevolent end users.

Can we construct a monitoring service that is trust-free? For example, we may punish the monitoring service providers if they fail to defend the states for the users. However, when delving into this idea, we immediately see some caveats that render this approach impractical. How much penalty should monitoring service providers pay? Ignoring the frictions, the total penalty bond for monitoring service providers should be equal to the largest potential loss incurred to the party that went offline. 

This effectively doubles the liquidity requirement for an off-chain network because whenever someone goes offline, in addition to the existing locked liquidity in channels or fraud-proof bond in sidechains, monitoring service providers also need to lock up the same amount of liquidity as penalty deposits. 

Worse, the monitoring service providers need to retain different assets for different monitoring tasks and things can get really complicated when the involved states are complex and multiple assets classes are in play. Sometimes, there is not even a straightforward translation from state to the underlying value, given all the complex state dependency for generalized state channels. 

Even if there is enough liquidity, the ``insurance" model here is really rigid: it is basically saying that you get $X\%$ back at once if the monitoring service providers fail to defend your states. If you choose a large value of $X$, it can become really expensive due to the additional liquidity locking, but if you choose a small value of $X$, it can become really insecure. 

On top of these disadvantages, it is unclear how the price of state monitoring services should be determined as market information is still segregated with low efficiency. This low efficiency and the per-party bond on heterogeneous assets will further cause complicated on-chain and off-chain interactions with monitoring services and smash the usability of any off-chain platform. There are more issues, but above are already bad enough.

To solve these issues, we propose State Guardian Network (SGN). State Guardian Network is a special compact side chain to guard off-chain states when users are offline. \tokenname token holders can stake their \tokenname into SGN and become state guardians. Before a user goes offline, she can submit her state to SGN with a certain fee and ask the guardians to guard her state for a certain period of time. A number of guardians are then randomly selected to be responsible for this state based on state hash and the ``responsibility score". The detailed rules for selecting the guardians are as follows.
\begin{itemize} 
\item{(\textbf{State guarding request}). A state guarding request is a tuple $\eta_i=(s_i, \ell_i, d_i)$ where $s_i$ is the state that should be guarded, $\ell_i$ is the amount of service fee paid to guardians and $d_i$ is the duration for which this state should be guarded.}
\item{(\textbf{Responsibility Score}). The responsibility score of a state guarding request $\eta_i$ is calculated as:
\[
\gamma_i  = \frac{\ell_i}{d_i}.
\]
A user's Responsibility Score is essentially the income flow generated by this user to the SGN.
}
\item{(\textbf{Number of guardian stakes}). Given a set of outstanding state guarding request $\mathcal{R}=\{\eta_1,\cdots, \eta_m\}$, the number of \tokenname at stake for each request $\eta_i\in\mathcal{R}$ is
\[
n_i = \frac{\gamma_i}{\sum\limits_{j=1}^m \gamma_j} K,
\]
where $K$ the total number of \tokenname stakes that guardians stake in the SGN.
In other words, the amount of responsible \tokenname staked is proportional to the ratio between this request’s responsibility score to the sum of all outstanding states’ responsibility scores.
}

\item{(\textbf{Assignment of guardian stakes}). Given a state guarding request $\eta_i$, let $h_i$ be the hash value for the corresponding state $s_i$ (e.g., Keccak256 hash). Each \tokenname stake $k$ is associated with an ID $p_k$ (which is also a hash value). Let $\delta(g_1, g_2)$ be the distance between two hash values $g_1$ and $g_2$ (e.g., the distance measure used in Chord DHT \cite{chord}). Then \tokenname stakes are sorted in ascending order by their distance to the hash value $h_i$. Suppose that $\delta(p_1,h_i)\le \delta(p_2,h_i)\le \cdots \le \delta(p_K,h_i)$ (ties are broken randomly). The first $n_i$ \tokenname stakes that have the smallest distance are selected, and the corresponding stake owner will become the state guardian for this request.}
\end{itemize}

\noindent (\textbf{State Guarding Service Fee Distribution}). For each state guarding request $\eta_i=(s_i, \ell_i, d_i)$,  the attached service fee $\ell_i$ is distributed to state guardians according to the following rule. For each state guardian $j$, let $z_{j}$ be the amount of his/her staked \tokenname that were selected for this state guarding request. Then the service fee that guardian $j$ gets from state guarding request $\eta_i$ is
\[
q_{j} = \frac{z_{j}\times \ell_i}{n_i}.
\]
Note that each staked \tokenname has the same probability of being selected for a state guarding request. As a result, from the view of an SG, the more \tokenname staked in SGN, the more of such SG's stakes will be selected in expectation (i.e., the value of $z_{j}$ will be larger), thus the amount of  service fees that he will receive will increase. That affords \tokenname significant value as a membership to the SGN.

\vspace{2mm}

\noindent (\textbf{Security and Collusion Resistance}). Each guardian is assigned a dispute slot based on the settlement timeout. If the guardian fails to dispute its slot when it ought to, subsequent guardians can report the event and get the failed guardian's \tokenname stake. As a result, as long as at least one of the selected guardians are not corrupted and fulfills the job, an end user's state is always safe and available for dispute.

The SGN mechanism also brings in the following additional values.
\begin{itemize} 
\item \textbf{It does not require significant liquidity lock-up for guardians.} Guardians are only staking their \tokenname which can be used to guard arbitrary states regardless of the type/amount of the underlying value/tokens.

\item \textbf{It provides a unified interface for arbitrary state monitoring.} Regardless of whether the state is related to ETH, any ERC20 tokens or complicated states, the users would just attach a fee and send it to SGN. SGN does not care about the underlying states and involved value, and simply allocates the amount of \tokenname proportional to the fee paid to be responsible for the state.

\item \textbf{It enables simple interactions.} Users of Celer Network do not need to contact individual guardians and they only need to submit states to this sidechain. 

\item \textbf{Most importantly, it enables an entirely new and flexible state guarding economic dynamics.} Instead of forcing the rigid and opaque ``get $X\%$ back" model, SGN brings users a novel mechanism to ``get my money back in $X$ period of time" and an efficient pricing mechanism for that fluid insurance model. If all guardians at stake fail to dispute for a user, she will get the \tokenname stakes from these guardians as compensation. In steady state, \tokenname tokens that are staked in the SGN represent an incoming flow (e.g., earning $x$ Dai/second). Ignoring the cost of state monitoring and other frictions, when a user submits the state to SGN, she can choose explicitly how much \tokenname is ``covering" for her state by choosing fees paid per second (i.e., the responsibility score).

\end{itemize}

\subsubsection{Summary} 
Thinking systematically, cEconomy covers the full life-cycle of an off-chain platform. LiBA and PoLC mining are about bringing intermediary transactions off-chain in a low-barrier fashion. SGN is about securing the capability to bring most up-to-date states back on-chain when needed. As such, we believe cEconomy is the first comprehensive off-chain platform cryptoeconomics that brings new value and enables otherwise impossible dynamics.

\section{Conclusion}

Celer Network is a coherent technology and economic architecture that brings Internet-level scalability to existing and future blockchains. It is horizontally scalable, trust-free, decentralized and privacy-preserving. It encompasses a layered architecture with significant technical innovations on each layer. In addition, Celer Network proposes a principled off-chain cryptoeconomics design to balance tradeoffs made to achieve scalability. Celer Network is on a mission to fully unleash the power of blockchain and revolutionize how decentralized applications are built and used.

\section{Founding Team}


\bigskip\noindent\textbf{Dr. Mo Dong} received his Ph.D. from UIUC. His research focuses on learning based networking protocol design, distributed systems, formal verification and Game Theory. Dr. Dong led project revolutionizing Internet TCP and improved cross-continental data transfer speed by 10X to 100X with non-regret learning algorithms. His work was published in top conferences, won Internet2 Innovative Application Award and being adopted by major Internet content and service providers. Dr. Dong was a founding engineer and product manager at Veriflow, a startup specializes in network formal verification. The formal verification algorithms he developed is protecting networking security for fortune 50 companies. Dr. Dong is also experienced in applying Algorithmic Game Theory, especially auction theory, to computer system protocol designs. He has been teaching full-stack smart contract courses. He produces technical blogs and videos on blockchain with over 7000 subscribers.

\bigskip\noindent\textbf{Dr. Junda Liu} received his Ph.D. from UC Berkeley, advised by Prof. Scott Shenker. He was the first to propose and develop DAG based routing to achieve nanosecond network recovery (1000x improvement over state of art). Dr. Liu joined Google in 2011 to apply his pioneer research to Google’s global infrastructure. As the tech lead, he developed a dynamic datacenter topology capable of 1000 terabit/s bisection bandwidth and interconnecting more than 1 million nodes. In 2014, Dr. Liu became a founding member of Project Fi (Google’s innovative mobile service). He was the tech lead for seamless carrier switching, and oversaw Fi from a concept to a \$100M+/year business within 2 years. He was also the Android Tech Lead for carrier services, which run on more than 1.5B devices. Dr. Liu holds 6 US patents and published numerous papers in top conferences. He received BS and MS from Tsinghua University.

\bigskip\noindent\textbf{Dr. Xiaozhou Li} received his Ph.D. from Princeton University and is broadly interested in distributed systems, networking, storage, and data management research. He publishes at top venues including SOSP, NSDI, FAST, SIGMOD, EuroSys, CoNEXT, and won the NSDI'18 best paper award for building a distributed coordination service with multi-billion QPS throughput and ten microseconds latency. Xiaozhou specializes in developing scalable algorithms and protocols that achieve high performance at low cost, some of which have become core components of widely deployed systems such as Google TensorFlow machine learning platform and Intel DPDK packet processing framework. Xiaozhou worked at Barefoot Networks, a startup company designing the world’s fastest and most programmable networks, where he led several groundbreaking projects, drove technical engagement with key customers, and filed six U.S. patents.

\bigskip\noindent\textbf{Dr. Qingkai Liang} received his Ph.D. degree from MIT in the field of distributed systems, specializing in optimal network control algorithms in adversarial environments. He first-authored over 15 top-tier papers and invented 5 high-performance and highly-robust adversarial resistant routing algorithms that have been successfully applied in the industry such as in Raytheon BBN Technologies and Bell Labs. He was the recipient of Best Paper Nominee at IEEE MASCOTS 2017 and Best-in-Session Presentation Award at IEEE INFOCOM 2016 and 2018.

\newpage
\bibliographystyle{tfs}
\bibliography{citations}


\end{document}